\documentclass{sig-alternate-05-2015}
\usepackage{amsmath}
\usepackage{amsfonts}
\usepackage{amssymb}
\usepackage{multirow}
\usepackage{hyperref}
\usepackage{algorithm}
\usepackage{algpseudocode}
\usepackage{graphicx}
\usepackage{url}
\usepackage[usenames]{color}
\usepackage{colortbl}
\usepackage{subcaption}
\usepackage{wasysym}
\usepackage{enumerate}
\usepackage{enumitem}
\usepackage{tabularx}
\usepackage{mathtools}
\usepackage{centernot}
\makeatletter
\algnewcommand{\LineComment}[1]{\Statex \hskip\ALG@thistlm \(\triangleright\) #1}
\newcommand{\RNum}[1]{\uppercase\expandafter{\romannumeral #1\relax}}
\newcommand{\vast}{\bBigg@{3}}
\newcommand{\Vast}{\bBigg@{4}}

\begin{document}

\graphicspath{{figures/}}
\title{DPHMM: Customizable Data Release with Differential Privacy via Hidden Markov Model}

\numberofauthors{6}


\author{%
{Yonghui Xiao{\small$~^{\#1}$}\titlenote{Part of the work was done when the author was at Samsung Research America, CA, USA.}, Yilin Shen{\small$~^{\dagger 2}$}, Jinfei Liu{\small$~^{\#3}$}, Li Xiong{\small$~^{\#4}$}, Hongxia Jin{\small$~^{\dagger 5}$}, Xiaofeng Xu{\small$~^{\#6}$} }%
\vspace{1.6mm}\\
\fontsize{10}{10}\selectfont\itshape
$~^{\#}$MathCS Department, Emory University, Atlanta, GA, USA\\
\fontsize{10}{10}\selectfont\itshape
$~^{\dagger}$Samsung Research America, San Jose, CA, USA\\
\fontsize{9}{9}\selectfont\ttfamily\upshape
\{$^{1}$yonghui.xiao, $^{3}$jliu253, $^{4}$lxiong, $^{6}$xxu37\}@emory.edu\ \ \{$^{2}$yilin.shen, $^{5}$hongxia.jin\}@samsung.com
}

\maketitle
\begin{abstract}
%
Hidden Markov model (HMM) has been well studied and extensively used.
In this paper, we present DPHMM ({Differentially Private Hidden Markov Model}), an HMM embedded with a private data release mechanism, in which the privacy of the data is protected through a graph.
Specifically,
we treat every state in Markov model as a node, and use a graph to represent the privacy policy, in which ``indistinguishability'' between states is denoted by edges between nodes.
Due to the temporal correlations in Markov model,
we show that the graph may be reduced to a subgraph with disconnected nodes, which become unprotected and might be exposed.
To detect such privacy risk, we define sensitivity hull and degree of protection based on the graph to capture the condition of information exposure.
Then to tackle the detected exposure, we study how to build an optimal graph based on the existing graph.
We also implement and evaluate the DPHMM on real-world datasets, showing that privacy and utility can be better tuned with customized policy graph.
\end{abstract}

\newtheorem{theorem}{Theorem}[section]
\newtheorem{lemma}{Lemma}[section]
\newtheorem{definition}{Definition}[section]
\newtheorem{corollary}{Corollary}[section]
\newtheorem{observation}{Observation}[section]
\newtheorem{proposition}{Proposition}[section]
\newtheorem{example}{Example}[section]
\newtheorem{fact}{Fact}[section]
\section{Introduction}
As information is widely shared and frequently exchanged in the big-data era,
data-owners' fear of privacy breach continues to escalate.
For instance, $78\%$ smartphone users among $180$ participants in a survey \cite{Fawaz:2014:CCS} believe that apps accessing their location pose privacy threats.
On the other hand, data collected from individual users can be of great value for both academic research and society, e.g. for purposes like data mining or social studies.
To release such data, private information must be retained.
As a result, private data release has drawn increasing research interest.

Markov model has been extensively used as a standard data model.
For example,
to analyze the web navigation behavior of users, the transitions between web-pages can be described through Markov model \cite{Deshpande:2004:SMM:990301.990304};
to analyze the moving patterns of users, Markov model (e.g., in Figure \ref{Figure-example-02} a user moves among $6$ locations) is also commonly adopted \cite{Quantifying-location-privacy-SP2011,MaskIt-SIGMOD12}.
To preserve the privacy in Markov model, the true state (e.g. the true webpage the user is browsing or the true location of a moving user) must be protected before the data is used or released. 

In this paper, we study the problem of private data release in Markov model. 
First, the true state that changes by Markov transition should be hidden from (not observable to) adversaries. Hence it is an HMM.
Second, different from the traditional HMM where the emission probabilities governing the distribution of the observed variables are given, we embed a private data release mechanism in HMM to determine the emission probabilities for privacy protection.
Given a function of the state, 
our goal is to release the answer of the function with the private data release mechanism at each timestamp.
Figure \ref{Figure-HMM} shows our problem.
\begin{figure}[t]
\centering
\begin{subfigure}{0.5\textwidth}
\includegraphics[width=8.5cm]{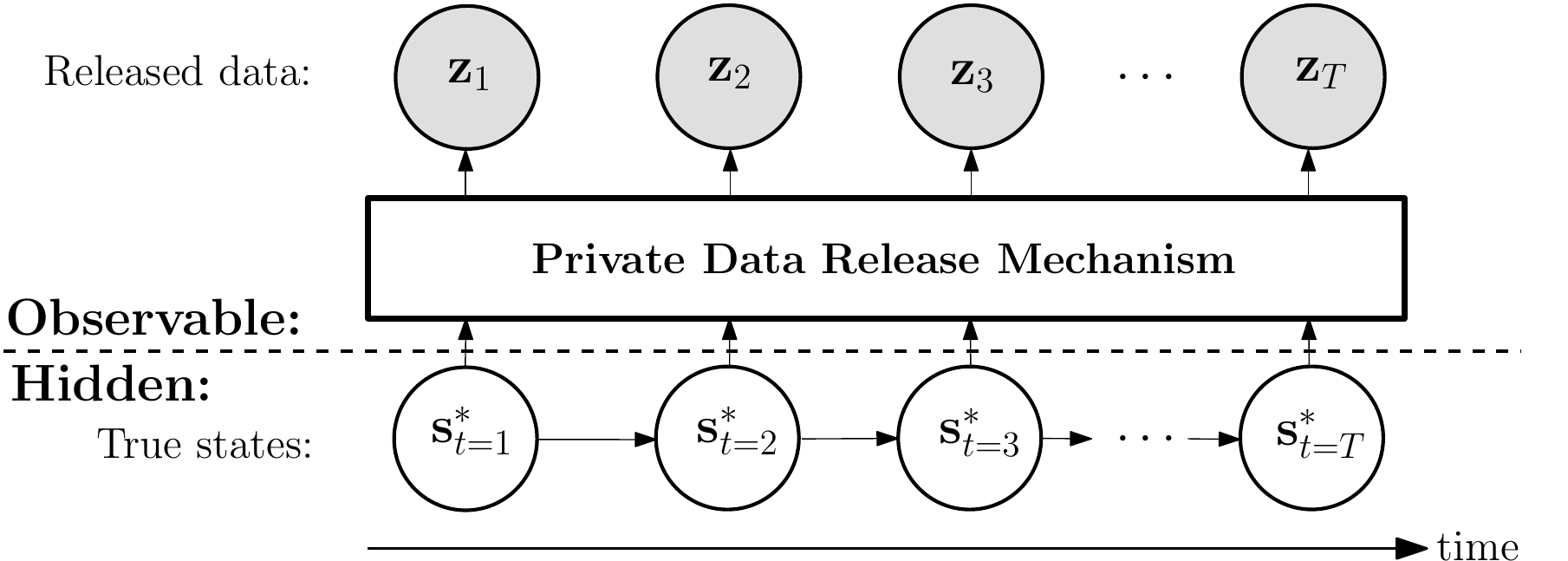}
\end{subfigure}
\caption{{\small Private data release mechanism embedded in HMM}}
\label{Figure-HMM}
\end{figure}

To design the private data release mechanism, 
there are two major difficulties.
\begin{itemize}
\item
How to tune the trade-off between privacy and utility with customizable privacy policy?
Most privacy notions in the literature only work in their specific problem settings, and lack the flexibility of trading-off privacy and utility. The state-of-art Blowfish framework \cite{Blowfish-SIGMOD14}, however, was proposed in the statistical database context, and cannot be directly adopted in Markov model.
\item
How to design a privacy notion under the temporal correlations in Markov model? Because most privacy notions proposed so far only focus on privacy models in static scenarios, they are vulnerable against inference attacks with temporal correlations in Markov model.
\end{itemize}
Next we explain the above difficulties in details. We first briefly introduce Blowfish framework. Then we show how several challenges emerge when adapting Blowfish framework in Markov model.
%



%

\begin{figure}[t]
\centering
\begin{subfigure}{0.22\textwidth}
\centering
\begin{tabular}{|c|c|c|}
\hline
$\textbf{s}_1$ & Alice& cancer \\\hline
$\textbf{s}_2$& Alice& asthma\\\hline
$\textbf{s}_3$ &Bob & cancer \\\hline
$\textbf{s}_4$ & Bob & diabetes\\\hline
$\textbf{s}_5$ &  Chad & cancer\\\hline
$\textbf{s}_6$ &  Chad & diabetes\\\hline
\end{tabular}
\caption{}
\label{Figure-example-TBlowfish}
\end{subfigure}
\begin{subfigure}{0.12\textwidth}
\centering
\includegraphics[width=1.8cm]{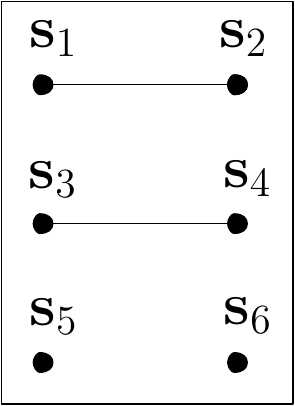}
\caption{}
\label{Figure-example-GBlowfish}
\end{subfigure}
\begin{subfigure}{0.12\textwidth}
\centering
\includegraphics[width=1.8cm]{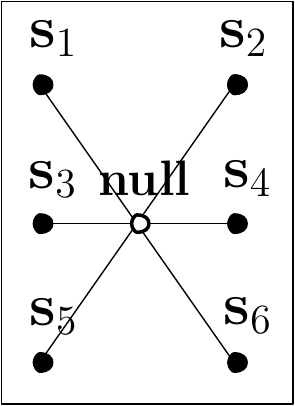}
\caption{}
\label{Figure-example-GBlowfish2}
\end{subfigure}
\caption{{\small (a): a table showing patients' diseases with each row being a secret; (b): a policy graph of bounded Blowfish; (c): a policy graph of unbounded Blowfish.}}
\label{Figure-example-Blowfish}
\end{figure}

%

\vspace{2mm}
\noindent{\bf Blowfish Privacy.}
Customizable privacy framework, Blowfish privacy \cite{Blowfish-SIGMOD14,Kifer-2012-pufferfish}, has been studied in statistical database context.
To tune the privacy and utility, it uses a policy graph  where a node represents a secret, and an edge represents indistinguishability between the two connected nodes.
For example, Figure \ref{Figure-example-TBlowfish} is a patients' table where each row is a secret indicating the patient's disease e.g., secrets $\textbf{s}_3$ and $\textbf{s}_4$ are ``Bob has cancer'' and ``Bob has diabetes'' respectively.
For bounded Blowfish privacy, it uses a policy graph to enforce the indistinguishability between the secrets, which can be regarded as \emph{edge protection} in the graph.
For instance,
 Figure \ref{Figure-example-GBlowfish} ensures adversaries cannot distinguish whether Bob has cancer or diabetes by connecting $\textbf{s}_3$ and $\textbf{s}_4$.
For unbounded Blowfish privacy, it uses a policy graph to disguise \emph{the existence of secrets}. For example, Figure \ref{Figure-example-GBlowfish2} connects all secrets to a ``null'' node, which represents the non-existence of these nodes. Thus the adversaries cannot know whether a secret is real or not. Furthermore, the bounded and unbounded Blowfish privacies can also be combined in one graph by adding the null node into the graph of bounded Blowfish.

Although the privacy customization of Blowfish is intuitive, the overall protection of Blowfish, which can be problematic in the following examples, has not been fully studied. 
\begin{itemize}
\item
Are secrets $\textbf{s}_5$ and $\textbf{s}_6$ also protected in Figure \ref{Figure-example-GBlowfish}?
Since they are disconnected in the graph,
%
 for their protections, is it necessary to connect them to a null node, or connect them to other nodes (and which)? 
\item
If $\{\epsilon,G\}$-Blowfish privacy is preserved where $G$ is the graph in Figure \ref{Figure-example-GBlowfish}, what is the privacy guarantee for all the secrets $\{\textbf{s}_1\sim\textbf{s}_6\}$, 
i.e., how to quantify Blowfish privacy in terms of differential privacy?
\end{itemize}
We will answer above questions in
Example \ref{example-Blowfish1}, followed with theoretical result.

\vspace{2mm}
\noindent{\bf Markov Model.} Markov model has been studied with differential privacy in existing works.
Chatzikokolakis et al. \cite{ChatzikokolakisPS14} and Fan et al. \cite{liyue-2014-www} used Markov model for improving utility of released location traces or web browsing activities, but did not consider the inference risks under the temporal correlations of the Markov model.
Xiao et al. \cite{LocPriv14-arXiv} studied how to protect a user's location, described through Markov model, in a set of possible locations (states) where the user might appear.
%
 However, such set of possible states could be either too big or too small for the user. In reality, the nature of privacy is determined by personal information. Hence it is necessary to  customize privacy protection for personal demands.


\begin{figure}[t!]
\centering
\begin{subfigure}{0.23\textwidth}
\centering
\includegraphics[width=4cm]{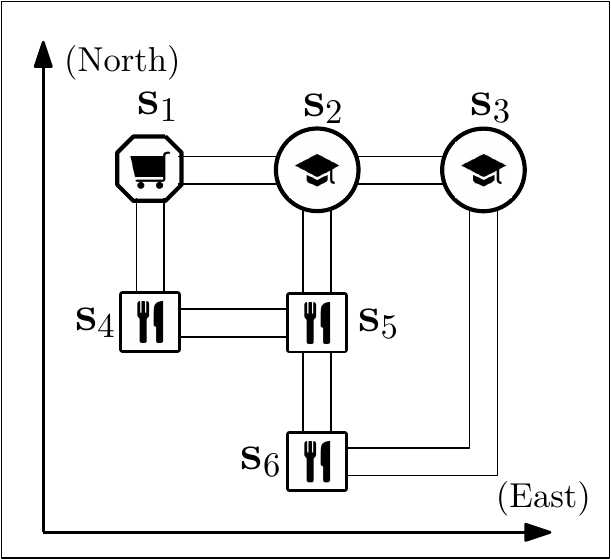}
\caption{}
\label{Figure-example-02}
\end{subfigure}
\begin{subfigure}{0.23\textwidth}
\centering
\includegraphics[width=4cm]{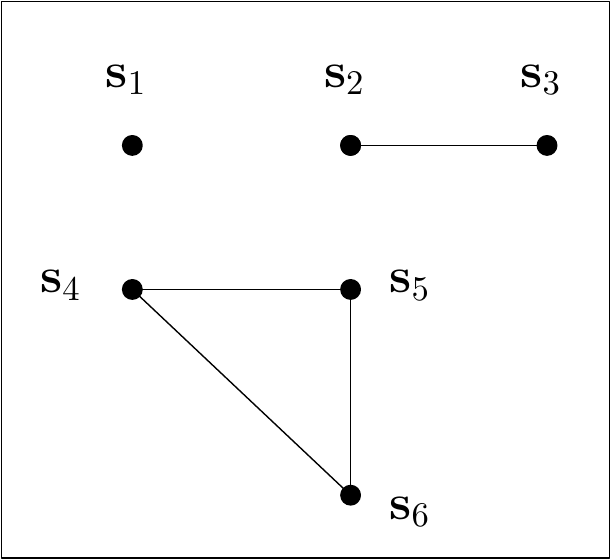}
\caption{}
\label{Figure-example-03}
\end{subfigure}
\\
\begin{subfigure}{0.23\textwidth}
\centering
\includegraphics[width=4cm]{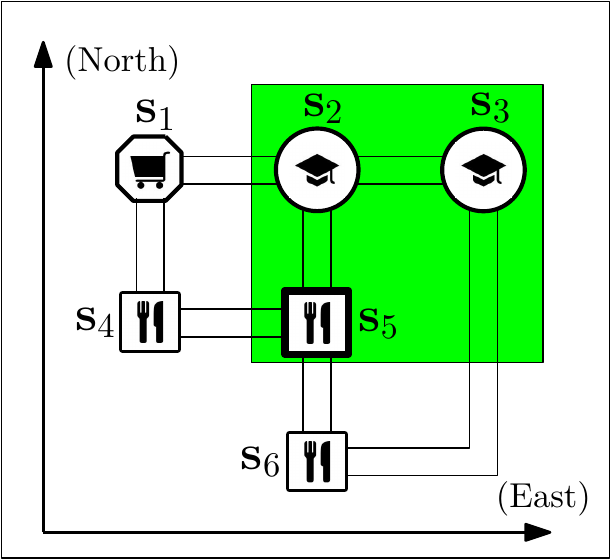}
\caption{}
\label{Figure-example-04}
\end{subfigure}
\begin{subfigure}{0.23\textwidth}
\centering
\includegraphics[width=4cm]{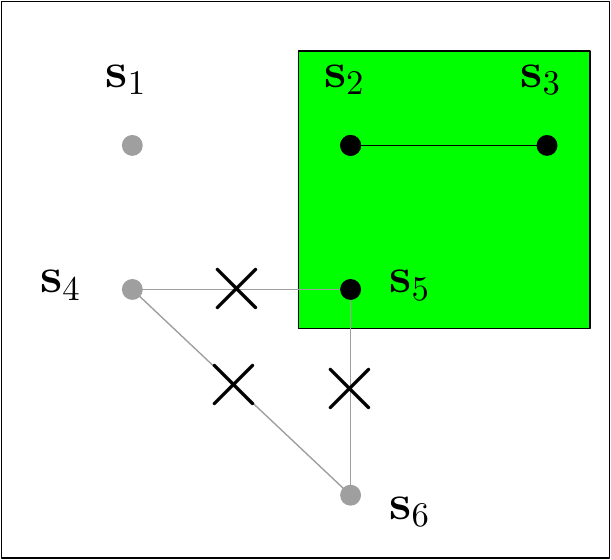}
\caption{}
\label{Figure-example-05}
\end{subfigure}
\caption{{\small Running example. (a): protecting a state in its category (the octagon, circles or squares); (b): the policy graph connecting all nodes in a category; (c): an adversary estimated that the true location can only be   $\{\textbf{s}_2,\textbf{s}_3,\textbf{s}_5\}$; (d): the reduced graph from (b) with the constraint in (c).}}
\label{Figure-example-customization}
\end{figure}

We can potentially apply Blowfish privacy in Markov model.
For example, Figure \ref{Figure-example-02} shows the Markov model of a moving user with $6$ states, denoted by $\{\textbf{s}_1,\cdots,\textbf{s}_6\}$.
If the user prefers to hide her state in $3$ categories, i.e., cafeteria, school and grocery (the octagon, circle and square in Figure \ref{Figure-example-02}), the privacy customization can be achieved by the graph in Figure \ref{Figure-example-03}.
Then if the user is at state $\textbf{s}_5$, the graph ensures that $\{\textbf{s}_4,\textbf{s}_5,\textbf{s}_6\}$ are indistinguishable.
%

Unfortunately, the policy graph may be reduced under the temporal correlations in Markov model.
For instance, assume the user moved from $\textbf{s}_1$ to $\textbf{s}_5$.
If an adversary infers by temporal correlations that the true state can only be $\{\textbf{s}_2,\textbf{s}_3,\textbf{s}_5\}$, the shaded area in Figure \ref{Figure-example-04}, is the graph in Figure \ref{Figure-example-03} still applicable?
In this case, although $\textbf{s}_5$ is connected to $\textbf{s}_4$ and $\textbf{s}_6$, the adversary can eliminate $\textbf{s}_4$ and $\textbf{s}_6$ with the knowledge (constraint). In consequence, the original edges $\overline{\textbf{s}_4\textbf{s}_5}$ and $\overline{\textbf{s}_5\textbf{s}_6}$ disappear, as shown in Figure \ref{Figure-example-05}. Then the following questions arise: is $\textbf{s}_5$ still protected? If not, how to re-generate a new graph to protect $\textbf{s}_5$ based on the current graph?
We will answer these questions in Examples \ref{example-sh}, \ref{example-exposure} and \ref{example-mpg}.

%
%

Another challenge of directly applying Blowfish privacy is the constraint type. In Blowfish framework, the constraints are deterministic, which leads to the NP-hard complexity \cite{Blowfish-SIGMOD14}. Whereas the constraints in Markov model are probabilistic. For example, in Blowfish framework, Bob can have cancer and diabetes at the same time, or no disease at all. However, in Markov model, there has to and can only exist ONE state,
which means the existence of one state excludes all other states.
Such rigid constraints pose higher privacy risk than in Blowfish framework.

At last, the long-term privacy protection after releasing a sequence of data should also be considered. For instance, assume the user moved from $\textbf{s}_1$ to $\textbf{s}_5$, and the real sequence is $\{\textbf{s}_1,\textbf{s}_2,\textbf{s}_5\}$. If other possible sequences can also be estimated by temporal correlations, like $\{\textbf{s}_1,\textbf{s}_4,\textbf{s}_5\}$, $\{\textbf{s}_1,\textbf{s}_2,\textbf{s}_3\}$, then what is the long-term protection for the these sequences?

\subsection{Contributions}
First, we propose a rigorous and customizable DPHMM notion by extending the Blowfish privacy \cite{Blowfish-SIGMOD14}. Specifically, we treat every state in Markov model as a node, and construct a graph, in which edges represent  ``indistinguishability'' between the connecting nodes, to represent the privacy policy.
In this way,
the DPHMM notion guarantees that the true state is always protected in its connecting ``neighbors''.

Second, we formally analyze the privacy risk under the constraint of temporal correlations.
We show  that the original graph may be reduced to a subgraph under the constraint, possibly with disconnected nodes.
To detect the information leakage of the disconnected nodes,
we define sensitivity hull and degree of protection (\textsc{DoP}) based on the graph to capture the protectability of a graph
(if a graph is not protectable, then the disconnected nodes will be exposed).
We also quantify the overall protection of Blowfish privacy in terms of differential privacy using the sensitivity hull.
In addition, we prove that Laplace mechanism \cite{Dwork-calibrating} is a special case of $K$-norm mechanism \cite{Geometry-Hardt-STOC10}, and provides no better utility than $K$-norm mechanism.

Third, we develop a data release mechanism to achieve DPHMM.
To tackle the detected information leakage, we study how to re-connect the disconnected nodes and find the optimal protectable graph based on the existing graph.
We also implement and evaluate the data release mechanism on real-world datasets, showing that privacy and utility can be better tuned with customized policy graph.

Fourth, we thoroughly study the privacy guarantee of DPHMM framework. Besides comparing DPHMM with other privacy notions,
we present the privacy composition results when multiple queries were answered over multiple timestamps.
%

\section{Related Works}
\subsection{Differential Privacy}
While differential privacy \cite{Dwork06differentialprivacy}  has been accepted as a standard notion for privacy protection,
most works used Laplace mechanism \cite{Dwork-calibrating} to release differentially private data.
Based on Laplace mechanism, Li et al. proposed Matrix mechanism \cite{Optimizing-PODS} to answer a batch of queries by factorizing a query matrix to generate a better ``strategy'' matrix that can replace the original query matrix.
Other mechanisms, such as Exponential mechanism \cite{McSherry-mechanism} and $K$-Norm mechanism \cite{Geometry-Hardt-STOC10}, were also proposed to guarantee differential privacy.
We refer readers to \cite{DPBench-SIGMOD16} for a  comparative study of the mechanisms. 
A variety of differentially private applications \cite{ChatzikokolakisPS14,liyue-2014-www,w-event-VLDB14} can also be found in literature.
%

Because the concept of standard differential privacy is not generally applicable,
several variants or generalizations of differential privacy, such as induced neighbors privacy \cite{kifer2011no},
have been proposed. Among these variants, Blowfish privacy \cite{Blowfish-SIGMOD14} is the first generic framework with customizable privacy policy. It defines sensitive information as secrets and known knowledge about the data as constraints. By constructing a policy graph, which should also be consistent with all constraints, Blowfish privacy can be formally defined. We extend Blowfish framework to Markov model, and quantify the overall protection of Blowfish privacy in both database context and Markov model.

%

The lower bound of differentially private query-answering was also investigated. Hardt and Talwar \cite{Geometry-Hardt-STOC10} proposed the theoretical lower bound for any differentially private mechanisms. To achieve the lower bound, they also studied $K$-Norm based algorithms to release differentially private data. In the query answering setting, $K$-Norm mechanism is optimal only when the sensitivity hull \cite{LocPriv14-arXiv} is in isotropic position.
In this paper, we extend the $K$-Norm mechanism by investigating the sensitivity hull $K$ in the new setting of DPHMM.

\subsection{Private Sequential Data}
To account for sequential data that changes over time, progresses were made under the assumption that data at different timestamps should be independent. Dwork et al. \cite{Dwork-continual-STOC10} proposed ``user-level'' and ``event-level'' differential privacy to answer count queries on binary bit data.  The approach is to use a binary tree technique to amortize Laplace noises to a range of nodes in the tree. Thus the noise magnitude becomes proportional to $log(T)$ where $T$ is the time period. The same result was also achieved in \cite{Chan-continual-2011}.
Kellaris et al. \cite{w-event-VLDB14} studied $w$-event privacy, which protects the continual events in $w$ consecutive timestamps by adjusting the allocation of privacy budget.
Overall, above works mainly focused on releasing data independently at each timestamp regardless of temporal correlations.


Temporal correlations were considered with Markov model in several recent works.
Several works considered Markov models for improving utility of released location traces or web browsing activities \cite{ChatzikokolakisPS14,liyue-2014-www}, but did not consider the inference risks when an adversary has the knowledge of the Markov model.
Xiao et al. \cite{LocPriv14-arXiv} studied how to protect the true location if a user's movement follows Markov model. The technique can be viewed as a special instantiation of DPHMM for a two-dimensional query (see Theorem \ref{theo-comparison-K-CCS} for details).
 In addition, DPHMM uses a policy graph to tune the privacy and utility in Markov model.

\section{Preliminaries and Problem Statement}
We denote scalar variables by normal letters, vectors by bold lowercase letters, and matrices by bold capital letters.
 Superscript $\textbf{x}^T$ is the transpose of a vector $\textbf{x}$;
 $\textbf{x}[i]$ is the $i$th element of $\textbf{x}$.
 Operators $\cup$ and $\cap$ denote union and intersection of sets; $|\cdot|$ denotes the number of elements in a set;
 $||\cdot||_p$ denotes $\ell_p$ norm; $\overline{\textbf{ab}}$ denotes a line connecting points $\textbf{a}$ and $\textbf{b}$.
Table \ref{tbl-symbols} summarizes some important symbols for convenience.
\begin{table}
\centering
\begin{tabular}{|c|c|}
\hline
$\mathcal{S}$ & domain of states in Markov model \\\hline
$\textbf{s}_i,\textbf{s}_j,\textbf{s}_k$& a state in Markov model\\\hline
$\textbf{s}^*$ & the true state \\\hline
$\textbf{z}$&the released (observed) answer\\\hline
$\textbf{p}_t^-$& prior probability (vector) at timestamp $t$\\\hline
$\textbf{p}_t^+$& posterior probability (vector) at timestamp $t$\\\hline
$\mathcal{C}$ & constraint (set)\\\hline
$K$ & sensitivity hull\\\hline
\end{tabular}
\caption{{\small Notation}}
\label{tbl-symbols}
\end{table}
\subsection{Differential Privacy}
Differential privacy protects a database by ensuring that neighboring databases generate similar output. W.l.o.g, we use $\textbf{x}\in\mathbb{R}^n$ to denote a database with $n$ tuples.
A query is a function $f(\textbf{x})$: $\textbf{x}\rightarrow \mathbb{R}^d$ that maps $\textbf{x}$ to $\mathbb{R}^d$.
We use $\textbf{z}$ to denote the answer of a query from a differentially private mechanism.
\begin{definition}[Differential Privacy]
\label{def-stadard-dp}
A randomized mechanism $\mathcal{A}()$ satisfies $\epsilon$-differential privacy if for any output $\textbf{z}$, $
\frac{Pr(\mathcal{A}(\textbf{x}_1)=\textbf{z})}{Pr(\mathcal{A}(\textbf{x}_2)=\textbf{z})}\leq e^{\epsilon}
$ where neighboring databases $\textbf{x}_1$ and $\textbf{x}_2$ satisfies
\begin{itemize}
\item
(Unbounded DP) $\textbf{x}_2$ can be obtained from $\textbf{x}_1$ by adding or removing a tuple.
\vspace{-2mm}
\item 
(Bounded DP) $\textbf{x}_2$ can be obtained from $\textbf{x}_1$ {by replacing a tuple}.
\end{itemize}
\end{definition}

\noindent{\bf Laplace Mechanism.}
Laplace mechanism is commonly used in literature. It is built on the $\ell_1$-norm sensitivity \cite{Dwork-calibrating}, defined as follows.
\begin{definition}[$\ell_1$-norm Sensitivity]
\label{def-standard-sensitivity}
For any query $f(\textbf{x})$: $\textbf{x}\rightarrow \mathbb{R}^d$, its $\ell_1$-norm sensitivity $S_f$ is the maximum $\ell_1$ norm of $f(\textbf{x}_1)-f(\textbf{x}_2)$ where $\textbf{x}_1$ and $\textbf{x}_2$ are any two neighboring databases.
\begin{align*}
S_f\coloneqq \mathop{max}\limits_{\textbf{x}_1,\textbf{x}_2\in \textrm{ neighboring databases}}||f(\textbf{x}_1)-f(\textbf{x}_2)||_1
\end{align*}
where $||\cdot||_1$ denotes the $\ell_1$ norm.
\end{definition}

A query can be answered by $f(\textbf{x})+Lap(S_f /\epsilon)$ to achieve $\epsilon$-differential privacy, where $Lap()\in\mathbb{R}^d$ are i.i.d. random noises drawn from Laplace distribution.

\vspace{2mm}
\noindent{\bf $K$-norm Mechanism.} $K$-norm, written as $||\cdot||_K$, is the (Minkowski) norm defined by convex body $K$ (i.e. $||\textbf{v}||_K=inf\{r>0:\textbf{v}\in rK\}$). Given any query $f$, its sensitivity hull $K$ can be derived \cite{LocPriv14-arXiv}. Then a differentially private answer of $f$ can be generated with $K$-norm mechanism as follows.

\begin{definition}[K-norm Mechanism \cite{Geometry-Hardt-STOC10}]
Given any function $f$ and its sensitivity hull $K$, a mechanism is $K$-norm mechanism if for any output $\textbf{z}$, the following holds:
\begin{align}
\label{eqn-pdf-K-Norm}
Pr(\textbf{z})=\frac{1}{\Gamma(d+1)\textsc{Vol}(K/\epsilon)}exp \left( -\epsilon||\textbf{z}-f(\textbf{x}^*)||_K \right)
\end{align}
where $f(\textbf{x}^*)$ is the true answer, $\Gamma()$ is Gamma function and $\textsc{Vol}()$ denotes volume.
\end{definition}

In this paper, we only focus on Laplace mechanism and $K$-norm mechanism for simplicity. Whereas our framework is applicable to any differentially private perturbation mechanisms.

\subsection{Blowfish Privacy}
Unlike differential privacy which protects all neighboring databases together, Blowfish privacy only protects the connected secrets in its policy graph. 
Below we only show the definition of Blowfish neighbors. Then Blowfish privacy can be obtained by replacing the neighboring databases $\textbf{x}_1$ and $\textbf{x}_2$ in Definition \ref{def-stadard-dp} with the following $D_1$ and $D_2$.
\begin{definition}[Blowfish Neighbors \cite{Policy-Blowfish-VLDB15}]
\label{def-Blowfish-neighbors}
Given a graph $G$ and a set of constraints $\mathcal{C}$, two databases $D_1$ and $D_2$ are neighbors if they satisfy the constraint $\mathcal{C}$, and
\begin{itemize}
\item
(Unbounded Blowfish) $D_2$ can be obtained by adding  a tuple to or removing a tuple from $D_1$ if the tuple (secret) is connected to a ``null'' node in $G$.
\vspace{-1mm}
\item
(Bounded Blowfish) $D_1$ and $D_2$ only differ one tuple, whose values in $D_1$ and $D_2$ are connected in $G$.
\end{itemize}
\end{definition}
We can see that unbounded Blowfish protects the existence of secrets, and bounded Blowfish protects the edges (connected nodes) in the graph.

\subsection{Hidden Markov Model}
We denote the domain of states by $\mathcal{S}$,
$\mathcal{S}=\{\textbf{s}_1,\textbf{s}_2,\cdots, \textbf{s}_N\}$ where each $\textbf{s}_i$ is a unit vector with the $i$th element being $1$ and other $N-1$ elements being $0$.
We denote $\textbf{s}^*$ the true state at each timestamp.
For privacy protection, $\textbf{s}^*$ is unobservable to (hidden from) any adversaries.
Thus it is an HMM. 
At timestamp $t$, we use a vector $\textbf{p}_t\in [0,1]^{1\times N}$ to denote the probability distribution of true state.
Formally, \begin{align*}
\textbf{p}_t[i]=Pr(\textbf{s}^*_t=\textbf{s}_i)
\end{align*}
where $\textbf{p}_t[i]$ is the $i$th element in $\textbf{p}_t$ and $\textbf{s}_i\in\mathcal{S}$.
\begin{figure}
\centering
\includegraphics[width=8.4cm]{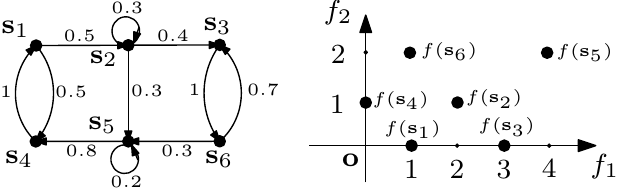}
\caption{{\small (left) a Markov model with transition probabilities; (right) its measurement query in Example \ref{example-f}. }}
\label{Figure-MC-example}
\end{figure}

\begin{example}[Running Example]
\label{example-tm}
The example in Figure \ref{Figure-example-02} is described by a random-walk Markov model in Figure \ref{Figure-MC-example} (left) where each state denotes a location on the map, if the true state at timestamp $t$ is $\textbf{s}_1$, then
$
\textbf{s}_t^*=\textbf{s}_1=
$[1 0 0 0 0 0],
%
%
$
\textbf{p}_t=
$[1 0 0 0 0 0].
\end{example}

\noindent{\bf Transition Probabilities.}
We use matrix $\textbf{M}\in [0,1]^{N\times N}$ to denote the transition probabilities with $m_{ij}$ being the probability of moving from state $i$ to state $j$. Given probability vector $\textbf{p}_{t-1}$, the probability at timestamp $t$ becomes
$\textbf{p}_{t}=\textbf{p}_{t-1}\textbf{M}$.
%
%

We will focus on first-order time-homogeneous Markov model in this paper with the understanding that our method can also be extended to high-order or time-heterogeneous  Markov model.

\vspace{2mm}\noindent{\bf Measurement Query.}
At each timestamp, a measurement query $f: \mathcal{S}\rightarrow \mathbb{R}^d$ about current state is evaluated.
We denote the space containing all possible outputs of $f$ by measurement space.

\begin{example}[Measurement Query]
\label{example-f}
Let $f:\mathcal{S}\rightarrow\mathbb{R}^2$ be two quantities about the true state in Figure \ref{Figure-MC-example}:
\begin{flalign*}
\hspace{2cm}
&f_1:\ \textrm{temperatue of current state}&\\
&f_2:\ \textrm{noise level of current state}&
\end{flalign*}
Then $f$ can be expressed as
$
f(\textbf{s})=
\left[
\begin{array}{cccccc}
1&2&3&0&4&1\\
0&1&0&1&2&2\\
\end{array}
\right]
\textbf{s}^T
$ 
where each column corresponds the answer of a state, e.g. $f(\textbf{s}_1)=[1,0]^T, f(\textbf{s}_2)=[2,1]^T$.
Above answer can be denoted in measurement space, as in Figure \ref{Figure-MC-example} (right).
\end{example}

%
%
\noindent{\bf Emission Probabilities.}
Emission probabilities $Pr(\textbf{z}_{t}|\textbf{s}^*_{t})$ denote the distribution of the observed variable $\textbf{z}_t$. In DPHMM, we design a private data release mechanism to answer the query $f$ with particular emission probabilities, which is the only difference between DPHMM and standard HMM.
%

\vspace{2mm}
\noindent{\bf Inference and Evolution.}
At timestamp $t$, we use $\textbf{p}_t^-$ and $\textbf{p}_t^+$  to denote the prior and posterior probabilities of an adversary about current state before and after observing $\textbf{z}_t$ respectively.
The prior probability can be derived by the (posterior) probability at  previous timestamp $t-1$
 and the Markov transition matrix as
$
\textbf{p}_{t}^-=\textbf{p}_{t-1}^+\textbf{M}
$.
The posterior probability can be computed using Bayesian inference as follows. For each state $\textbf{s}_i$:
\begin{align}
\textbf{p}_{t}^+[i]=Pr(\textbf{s}_{t}^*=\textbf{s}_i|\textbf{z}_{t})
=\frac{Pr(\textbf{z}_{t}|\textbf{s}_{t}^*=\textbf{s}_i)\textbf{p}_{t}^-[i]}{\mathop\sum\limits_{j}Pr(\textbf{z}_{t}|\textbf{s}_{t}^*=\textbf{s}_j)\textbf{p}_{t}^-[j]}
\label{eqn-posterior}
\end{align}

The inference of the true state at any timestamp can be efficiently computed by the forward-backward algorithm, which is also incorporated in our data release mechanism. 
Other standard HMM algorithms can also be directly used in DPHMM.

\subsection{Problem Statement}
Given an initial state (or probability) and a Markov model, our problem is to answer a measurement query $f:\mathcal{S}\rightarrow\mathbb{R}^d$ at each timestamp under the HMM assumptions.
First, the Markov model can be known to any adversaries. Second, all the previously released answers (observable) can be accessed by adversaries to make inference about the true state. Third, the data release mechanism is transparent to adversaries.
The released answer $\textbf{z}_t$ should have the following properties:

(1) it guarantees a privacy notion to protect the true state;

(2) it minimizes the error, measured by the $\ell_2$ distance between the released answer and the true answer $f(\textbf{s}^*)$:
\begin{align}
\label{equation-util}
\textsc{Error}=\sqrt{\mathbb{E}||\textbf{z}_t-f(\textbf{s}_t^*)||_2^2}
\end{align}

(3) the privacy-utility trade-off can be customized for various privacy requirements.
%

\vspace{2mm}
\noindent{\bf Learning the Markov Model.}
A Markov model can be learned from publicly available data or perturbed personal data using standard methods, such as EM algorithm. Even if an adversary can obtain such a model, we still need to protect the true state.
In the DPHMM, we assume the Markov model has been learned, and is also known to any adversaries.

\vspace{2mm}
\noindent{\bf Incomplete Model.}
Depending on the power of adversaries, an incomplete (inaccurate) Markov model can be used by adversaries. In this case, the privacy is still guaranteed while the inference result may be downgraded for the adversary (Appendix \ref{sec-AK}).


\section{Privacy Definition}
To derive the meaning of DPHMM, we extend Blowfish privacy from \cite{Blowfish-SIGMOD14,Policy-Blowfish-VLDB15}. Related privacy notions are also discussed in this section.
%
%
%

\subsection{Probabilistic Constraint}
\label{sec-Privacy-constraint}


A main difference between Blowfish framework and our framework is the constraint type. In Blowfish framework, the constraints are deterministic, which leads to the NP-hard complexity \cite{Blowfish-SIGMOD14}. While in Markov model, the constraints are probabilistic.
It means the probabilities of states can be known to adversaries.
At any timestamp $t$, the prior probability $\textbf{p}_t^-$ can be derived as
\begin{align*}
\textbf{p}_t^-[i]=Pr(\textbf{s}_t^*=\textbf{s}_i|\textbf{z}_{t-1},\cdots,\textbf{z}_1)
\end{align*}
Clearly, with $\textbf{p}_t^-$ the states can be divided into two sets: $\textbf{p}_t^-=0$ and $\textbf{p}_t^->0$. Such probabilistic constraints result in the following consequences. 
\begin{itemize}
\item
For the non-existing states ($\textbf{p}_t^-[i]=0$),
the unbounded Blowfish privacy is meaningless.  For example, if an adversary knows ``Bob does not have cancer'', it is not necessary to pretend ``Bob might have cancer'' any more.
\item
For the possible states ($\textbf{p}_t^-[i]>0$), unbounded Blowfish becomes bounded Blowfish privacy automatically, explained as follows.  In the definition of unbounded Blowfish, the neighbors mean $\textbf{s}_i$ exists or not. When $\textbf{s}_i$ does not exist, there has to exist another state \footnote{It means the ``null'' node is invalid in Markov modoel.}. Hence it becomes bounded Blowfish neighbors. This is different from traditional Blowfish, in which a database without any secret is still valid.
\item
Bounded Blowfish privacy only holds for the possible states because all edges connecting the non-existing states disappear in the policy graph.
\end{itemize}


Without ambiguity, we define the constraint of Markov model as the set of states with $\textbf{p}_t^->0$.
\begin{definition}[Constraint]
\label{def-MM-constraint}
Let $\textbf{p}^-_t$ be the prior probability at timestamp $t$. Constraint $\mathcal{C}_t$ consists of all states satisfying the constraint $\textbf{p}^-_t[i]>0$.
\begin{equation*}
\mathcal{C}_t\coloneqq \{\textbf{s}_i| \textbf{p}^-_t[i]>0, \forall \textbf{s}_i\in \mathcal{S} \}
\end{equation*}
\end{definition}

In conclusion, we focus on the bounded Blowfish, which means a state is mixed with other states in the graph,
under the constraint $\mathcal{C}_t$. 



\subsection{Policy Graph}
\label{sec-policy-graph}
\noindent{\bf Policy Graph without Constraint.}
We first study the problem in the whole domain $\mathcal{S}$ without any constraint.
Given the true state $\textbf{s}^*$ at a timestamp, a user may prefer to hide $\textbf{s}^*$ in a group of candidate states, denoted by $\mathcal{N}(\textbf{s}^*)$ as neighbors of $\textbf{s}^*$ where $\mathcal{N}(\textbf{s}^*) \subseteq \mathcal{S}$.
Intuitively, the more neighbors a state has, the more privately it is protected.
For simplicity, we assume $\textbf{s}_i\in\mathcal{N}(\textbf{s}_i)$ for all states $\textbf{s}_i$ because it is straightforward that $\textbf{s}_i$ is hidden in its neighbor set $\mathcal{N}(\textbf{s}_i)$.

%
We can represent the privacy policy by a undirected graph where a node represents a state and an edge connects an indistinguishable pair of states.

\begin{definition}[Policy]
A policy is an undirected graph $G=(\mathcal{S},\mathcal{E})$ where
$\mathcal{S}$ denotes all states (nodes) and $\mathcal{E}$ represents indistinguishability (edges) between states.
\end{definition}

\begin{definition}[Neighbors]
Let $\textbf{s}$ be a state in $\mathcal{S}$. The neighbors of $\textbf{s}$, denoted by $\mathcal{N}(\textbf{s})$, is the set of nodes connected with $\textbf{s}$ by an edge, including $\textbf{s}$ itself.
\begin{align*}
\mathcal{N}(\textbf{s})\coloneqq\{\textbf{s}\}\cup\{\textbf{s}'|\overline{\textbf{s}\textbf{s}'} \in \mathcal{E}, \textbf{s}'\in \mathcal{S}\}
\end{align*}
\end{definition}


To better adjust utility and privacy for any particular applications, how to design policy graph is not a trivial task.
Below we present a few examples of policy graphs, some of which are from database context \cite{Blowfish-SIGMOD14}.  In DPHMM, we assume a policy graph is given.
%
%
%
\begin{itemize}[leftmargin=*]
\item
Complete protection. To thoroughly protect a sensitive state, we can connect it with all other states.
In this way all states are connected and it forms a complete graph, as shown in Figure \ref{Figure-G-complete}. However, with higher privacy level comes less utility. Such policy may result in useless output.
\begin{align*}
G_{cplt}\coloneqq\{G| \overline{\textbf{s}\textbf{s}'}\in \mathcal{E},\ \forall \textbf{s},\textbf{s}'\in \mathcal{S}\}
\end{align*}
\item
Categorical protection. A common method to balance privacy and utility is to partition (or cluster) states into categories. Then every state only needs to be protected in its category. If in a category all states are connected, then the graph becomes disjoint cliques.
Figure \ref{Figure-G-partition} shows such an example.
\begin{align*}
G_{categ}\coloneqq\{G|G=G_1+G_2+\cdots+G_q, \forall i,j,G_i\cap G_j=\varnothing\}
\end{align*}
\item
Utility oriented policy. To improve utility, we may consider the policy in the measurement space of query $f$. Figure \ref{Figure-G-local} shows an example where nodes are only connected if their answers are within $r$ distance ($\ell_2$ distance in this example).
\begin{align*}
G_{util}\coloneqq\{G|\overline{\textbf{s}\textbf{s}'}\in \mathcal{E}\ \textrm{iff}\ dist(f(\textbf{s}),f(\textbf{s}'))\leq r,\ \forall \textbf{s},\textbf{s}'\in \mathcal{S}  \}
\end{align*}
where $dist()$ is a distance function in measurement space.
\item
One-step transition protection.
To protect a one-step transition $\textbf{s}_i\to\textbf{s}_j$, we can
require all pairs of states $\textbf{s}_j$ and $\textbf{s}_k$ to be indistinguishable if they can be transited from the same previous state $\textbf{s}_i$.
For example, if transition probabilities are given in Example \ref{example-tm}, then $G_{trs}$ can be derived as Figure \ref{Figure-G-transit} by the following equation.
\begin{align*}
G_{trs}\coloneqq\{G|\overline{\textbf{s}_j\textbf{s}_k}\in\mathcal{E} \textrm{ iff } m_{ij}>0\textrm{ and }m_{ik}>0, \forall i,j,k  \}
\end{align*}
Note that with $G_{trs}$ even if $\textbf{s}_t^*$ were exposed, $\textbf{s}_{t+1}^*$ would still be protected. Hence $G_{trs}$ provides strong privacy guarantee.
\end{itemize}

\begin{figure}[t]
\centering
\begin{subfigure}{0.11\textwidth}
\centering
\includegraphics[width=2cm]{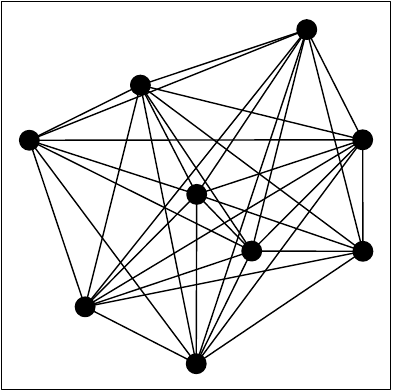}
\caption{{\small $G_{cplt}$}}
\label{Figure-G-complete}
\end{subfigure}
\begin{subfigure}{0.11\textwidth}
\centering
\includegraphics[width=2cm]{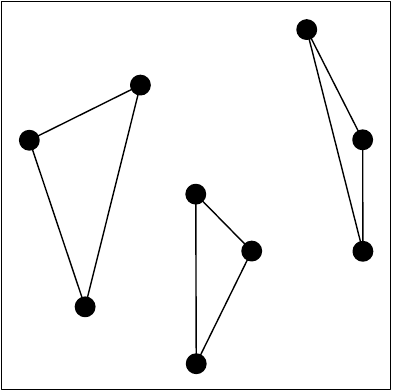}
\caption{{\small$G_{categ}$}}
\label{Figure-G-partition}
\end{subfigure}
\begin{subfigure}{0.11\textwidth}
\centering
\includegraphics[width=2cm]{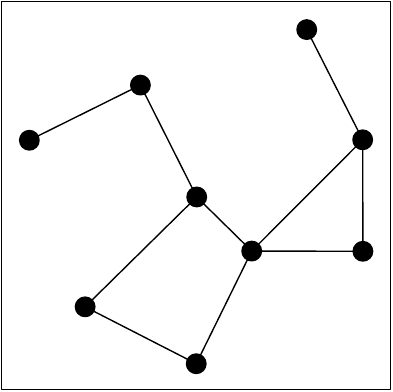}
\caption{{\small$G_{util}$}}
\label{Figure-G-local}
\end{subfigure}
\begin{subfigure}{0.11\textwidth}
\centering
\vspace{-0.01cm}
\includegraphics[width=2.02cm]{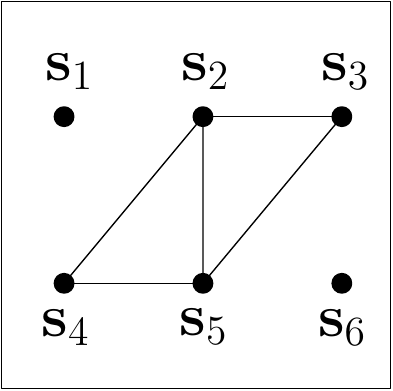}
\caption{{\small$G_{trs}$}}
\label{Figure-G-transit}
\end{subfigure}
\caption{{\small Examples of policy graphs without constraint. (a): complete protection; (b): categorical protection; (c): utility oriented policy; (d): transition protection for Example \ref{example-tm}.}
}
\label{Figure-Graph-Policy}
\end{figure}
%
%

\noindent{\bf  Policy Graph with Constraint.}
With the constraint $\mathcal{C}_t$ (Definition \ref{def-MM-constraint}), the policy graph $G$ has to be built on $\mathcal{C}_t$ at each timestamp $t$. Then the policy graph becomes a subgraph with the nodes in $\mathcal{C}_t$ and the residual edges in $G$, denoted by constrained policy graph $G\cap\mathcal{C}_t$. It is intuitive that with different $\mathcal{C}_t$ graphs may be different over time.
\begin{figure}[]
\centering
\begin{subfigure}{0.11\textwidth}
\centering
\includegraphics[width=2cm]{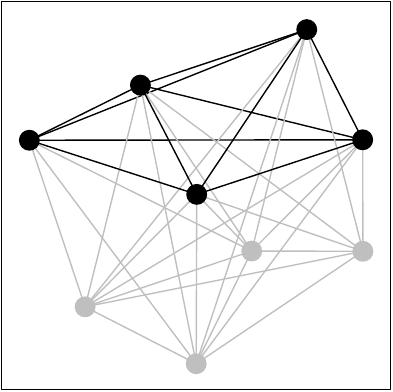}
\caption{{\small$G_{cplt}\cap\mathcal{C}_1$}}
\label{Figure-G-complete1}
\end{subfigure}
\begin{subfigure}{0.11\textwidth}
\centering
\includegraphics[width=2cm]{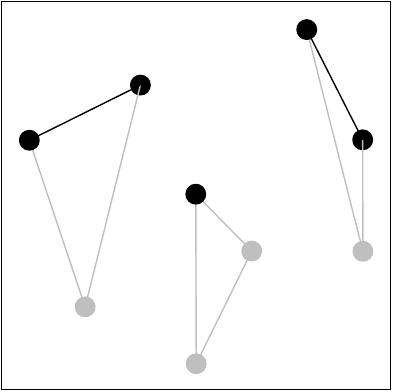}
\caption{{\small$G_{categ}\cap\mathcal{C}_1$}}
\label{Figure-G-partition1}
\end{subfigure}
\begin{subfigure}{0.11\textwidth}
\centering
\includegraphics[width=2cm]{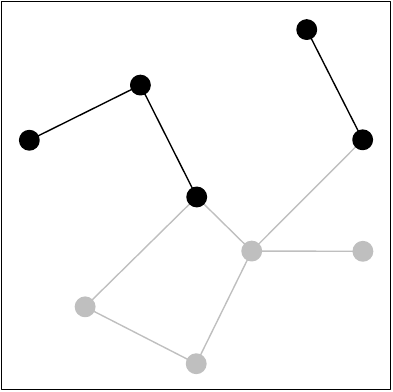}
\caption{{\small$G_{util}\cap\mathcal{C}_1$}}
\label{Figure-G-local1}
\end{subfigure}
\begin{subfigure}{0.11\textwidth}
\centering
\vspace{-0.01cm}
\includegraphics[width=2.02cm]{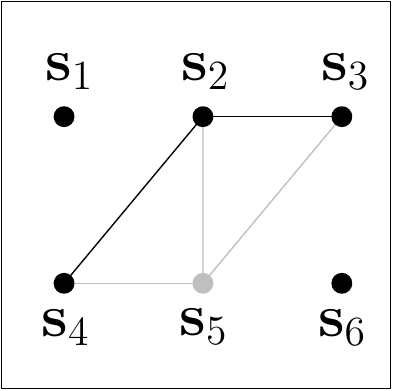}
\caption{{\small$G_{trs}\cap\mathcal{C}_2$}}
\label{Figure-G-transit1}
\end{subfigure}
\caption{{\small Policy graphs of Figure \ref{Figure-Graph-Policy} with constraints $\mathcal{C}_1$ and $\mathcal{C}_2$, denoted by the black points in the graphs.}
}
\label{Figure-Graph-Policy1}
\end{figure}
\begin{example}[Constrained Policy Graph]
Figure \ref{Figure-Graph-Policy1} shows the policy graphs of Figure \ref{Figure-Graph-Policy} with the constraint sets. The black points indicate the constraint sets. The gray points and their edges are removed from the original graph.
\end{example}

%
%
%

\subsection{DPHMM}
With policy graph $G$ and any constraint $\mathcal{C}_t$, $\{\epsilon,G,\mathcal{C}_t\}$-DPHMM can be defined as follows with the intuition that at any timestamp the true state cannot be distinguished from its remaining ``neighbors'' under the constraint.
\begin{definition}[$\{\epsilon,G,\mathcal{C}_t\}$-DPHMM]
\label{def-DPMC}
Let $G$ be the policy graph, $\mathcal{C}_t$ be the constraint at timestamp $t$.
An $\{\epsilon,G,\mathcal{C}_t\}$-DPHMM algorithm $\mathcal{A}()$ generates an output $\textbf{z}_t$ such that
for any $\textbf{z}_t$ and any state $\textbf{s}_j\in\mathcal{C}_t$, the following condition is satisfied:
\begin{align*}
\bigg\{
\begin{array}{ll}
\forall\textbf{s}_k\in\mathcal{N}(\textbf{s}_j)\cap\mathcal{C}_t,&\textrm{if $\textbf{s}_j$ is connected to $\textbf{s}_k$ in $G\cap\mathcal{C}_t$;}
\\
\exists \textbf{s}_k\in\mathcal{C}_t,&\textrm{if $\textbf{s}_j$ is disconnected in $G\cap\mathcal{C}_t$};
\end{array}
\end{align*}
\begin{align}
\label{eqn-DPMM}
e^{-\epsilon}\leq \frac{Pr(\mathcal{A}(\textbf{s}_j)=\textbf{z}_t)}{Pr(\mathcal{A}(\textbf{s}_k)=\textbf{z}_t)}\leq e^{\epsilon}
\end{align}
\end{definition}
In above definition, if $\textbf{s}_j$ is connected with any $\textbf{s}_k$ in $\mathcal{C}_t$, then $\textbf{s}_j$ and $\textbf{s}_k$ are indistinguishable by Equation (\ref{eqn-DPMM}); However,
if $\textbf{s}_j$ is disconnected, $\textbf{s}_j$ may be exposed \footnote{
The exposure consists of two scenarios: (1) an adversary knows $\textbf{s}_j$ is the true state. (2) an adversary knows $\textbf{s}_j$ is not the true state.
}.
To protect $\textbf{s}_j$ in this case, we have to connect $\textbf{s}_j$ to another node $\textbf{s}_k$ in $\mathcal{C}_t$ (such new graph is called protectable graph in Section \ref{sec-privacy-exposure}) to form a new edge of indistinguishability between $\textbf{s}_j$ and $\textbf{s}_k$. The user has the choice to specify which $\textbf{s}_k$ to use to protect $\textbf{s}_j$ \footnote{We do not hide the policy information (e.g. $\textbf{s}_j$ and $\textbf{s}_k$ are connected). DPHMM ensures that an adversary cannot distinguish whether $\textbf{s}_j$ or $\textbf{s}_k$ is the true state.}. We also discuss how to find the optimal $\textbf{s}_k$ in Section \ref{sec-complexity-OG}.

\subsection{Comparison with Other Definitions}
Among the variant definitions of differential privacy \cite{kifer2011no,Kifer-2012-pufferfish,Blowfish-SIGMOD14,LocPriv14-arXiv}
we briefly compare some closely related definitions as follows.

\vspace{2mm}
\noindent{\bf $\delta$-Location Set based Differential Privacy.}
\cite{LocPriv14-arXiv} defined differential privacy on a subset of possible states (locations) derived from Markov model. The indistinguishability is ensured among any two locations in the $\delta$-location set, which can be viewed as a new constraint. Thus it is a special case of DPHMM with complete graph.
\begin{theorem}
\label{theo-comparison-K-CCS}
$\delta$-location set based $\epsilon$-differential privacy \cite{LocPriv14-arXiv} is equivalent to $\{\epsilon,G_{cplt},\mathcal{C}_t'\}$-DPHMM where $G_{cplt}$ is a complete graph and 
 $
\mathcal{C}_t'=min\{\textbf{s}_i| \sum_{\textbf{s}_i}\textbf{p}_t^-[i]\geq 1-\delta\}
$.
\end{theorem}

%
%

\noindent{\bf Blowfish Framework.}
There are three differences between DPHMM and Blowfish framework.
(1)
The constraints in Blowfish are deterministic; while constraints in Markov model are probabilistic.
(2)
The graph in Blowfish is static; while in Markov model the graph can be reduced. When there are disconnected nodes in the reduced graph, privacy risk needs to be tackled.
(3)
We quantify the privacy guarantee of Blowfish in terms of differential privacy (Section \ref{sec-Blowfish5}).
\begin{theorem}
$\{\epsilon,G,\mathcal{C}_t\}$-DPHMM is equivalent to $\{\epsilon,\{\mathcal{S},\\\mathbb{G}_t,\mathcal{C}_t\}\}$-Blowfish privacy where $\mathcal{S}$ is the domain of states in a Markov model, $\mathbb{G}_t$ is the set of graphs satisfying the condition in DPHMM and $\mathcal{C}_t$ is the constraint.
\end{theorem}


\section{Privacy Risk}
Given a constrained policy graph $G\cap\mathcal{C}_t$, when a node $\textbf{s}_i$ is disconnected (without neighbors), one may conclude that $\textbf{s}_i$ will be disclosed. However, we show that this  may not be the case for Laplace mechanism or $K$-norm mechanism. The reason is that the perturbation is based on the sensitivity of a query, and the sensitivity may implicitly protect $\textbf{s}_i$ with other nodes.
In this section, we formalize the intuition, and define sensitivity hull and degree of protection (\textsc{DoP}) based on the constrained policy graph to analyze the privacy risk.
We also analyze the overall protection of Blowfish privacy with the sensitivity hull.
\subsection{Sensitivity Hull}
It has been shown that the standard $\ell_1$-norm sensitivity (in Definition \ref{def-standard-sensitivity}) exaggerates the sensitivity of differential privacy \cite{LocPriv14-arXiv}.
To capture the real sensitivity, we define sensitivity hull of graph $G$ and query $f$ using convex hull, denoted by $Conv()$.
Intuitively, it measures the ``maximum'' differences of the query results  on each pair of connected states (edges in the graph).

\begin{definition}[Sensitivity Hull]
\label{def-shull1}
Given a graph $G\\=(\mathcal{S},\mathcal{E})$, the sensitivity hull of a query $f$ is the convex hull of $\Delta f$ where $\Delta f$ is the set of $f(\textbf{s}_j)-f(\textbf{s}_k)$ for any connected nodes $\textbf{s}_j$ and $\textbf{s}_k$ in $G$.
\begin{flalign*}
\hspace{2.4cm}
K(G,f)&= Conv\left(\Delta f\right)&
\\
\Delta f&=\mathop\cup\limits_{\overline{\textbf{s}_j\textbf{s}_k}\in \mathcal{E}}
\left( f({\textbf{s}_j})-f({\textbf{s}_k}) \right)&
\end{flalign*}
\end{definition}

\begin{figure}
\centering
\begin{subfigure}{0.23\textwidth}
\centering
\vspace{0.65cm}
\includegraphics[width=3.8cm]{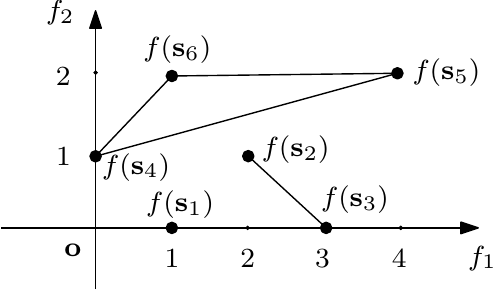}
\vspace{0.6cm}
\caption{}
\label{Figure-SHull2}
\end{subfigure}
\begin{subfigure}{0.24\textwidth}
\centering
\includegraphics[width=3.8cm]{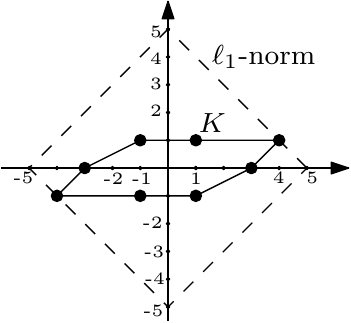}
\caption{}
\label{Figure-SHull}
\end{subfigure}
\caption{{\small Sensitivity hull without constraint. (a): the policy graph of Figure \ref{Figure-example-customization} in measurement space; (b): the $\ell_1$-norm sensitivity (dashed lines) and the sensitivity hull $K$ (solid lines).}}
\label{Figure-SHull-all}
\end{figure}

Without ambiguity, $\Delta f$ can also be denoted by a matrix in $\mathbb{R}^{d\times 2m}$ where each column is a point of $f(\textbf{s}_j)-f(\textbf{s}_k)$ and $m=|\mathcal{E}|$ is the number of edges in $G$. We show an example as follows.
\begin{example}[Sensitivity Hull]
\label{example-sh}
Given the query in Example \ref{example-f} and the graph in Figure \ref{Figure-example-03}, Figure \ref{Figure-SHull2} shows the policy graph without constraint in measurement space. The $\ell_1$-norm sensitivity (Definition 5.1 in \cite{{Blowfish-SIGMOD14}}) is $5$ because $||f(\textbf{s}_4)-f(\textbf{s}_5)||_1=5$. The dashed lines and solid lines show the $\ell_1$-norm sensitivity and sensitivity hull with the following $\Delta f$ in Figure \ref{Figure-SHull} respectively. Each column in $\Delta f$ denotes the query difference of two states, e.g., the first column $(-1,1)^T$ is $f(\textbf{s}_2)-f(\textbf{s}_3)$.
\begin{align*}
\Delta f=\left[
\begin{array}{cccccccc}
-1&1&-4&4&-1&1&3&-3\\
1&-1&-1&1&-1&1&0&0\\
\end{array}
\right]
\end{align*}
\end{example}
%
%

\vspace{2mm}
\noindent{\bf Computation.} The computation of sensitivity hull involves two steps, $\Delta f$ and $Conv()$, with $O(m^2)$ and
$O(mlog(m)\\+m^{\lfloor d/2 \rfloor})$ \cite{convexhull-alg1993} complexity respectively, where $m=|\mathcal{E}|$ is the number of edges in $G$, if the answer of $f$ is given. Therefore, the overall complexity is $O(m^2)$ for $d\leq 4$ and $O(m^{\lfloor d/2 \rfloor})$ for $d>4$. We skip the computation details because $Conv()$ has been well studied in computational geometry.

\vspace{2mm}
\noindent{\bf Discussion.}
As shown in  Figure \ref{Figure-SHull}, $\ell_1$-norm sensitivity is bigger than sensitivity hull.
Following this, we can further prove that Laplace mechanism is a special case of $K$-norm mechanism and provides  no better utility than $K$-norm mechanism. Thus we use $K$-norm mechanism as a unifying mechanism in the following analysis of this paper. 
\begin{theorem}
\label{theo-lap-K}
Laplace mechanism is a special case of $K$-norm mechanism when $K=K^{\Diamond}_{S_f}$ where
$K^{\Diamond}_{S_f}$ is the cross polytope $\{\textbf{x}\in\mathbb{R}^d: ||\textbf{x}||_1\leq S_f\}$ and $S_f$ is the $\ell_1$-norm sensitivity of $f$.
\end{theorem}
\begin{corollary}
\label{coro-lap-K}
Laplace mechanism provides no better utility than $K$-norm mechanism
because $K^{\Diamond}_{S_f}$ always contains sensitivity hull $K$.
\end{corollary}

\subsection{Privacy Risk}
\label{sec-privacy-exposure}
Under constraint, connectivity of policy graph is destructed. In the following example, we show that directly using existing data release methods may lead to exposure of disconnected nodes.
\begin{figure}[]
\centering
\includegraphics[width=8cm]{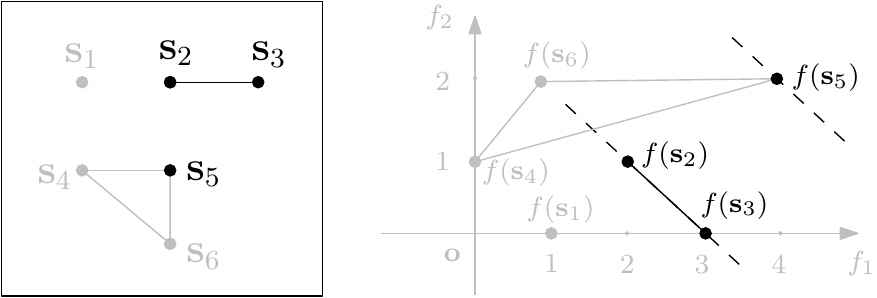}
\caption{{\small (left) constrained policy graph; (right) the query results in measurement space.}
}
\label{Figure-exposure}
\end{figure}
\begin{example}[Information Exposure]
\label{example-exposure}
Given the query in Example \ref{example-f} and the graph in Figure \ref{Figure-example-03}, assume at a timestamp $t$, the constraint set $\mathcal{C}_t=\{\textbf{s}_2,\textbf{s}_3,\textbf{s}_5\}$.
Then the constrained graph is shown in Figure \ref{Figure-exposure} (left). We use existing mechanisms, including Laplace based mechanisms and $K$-norm based mechanisms, to answer the query in Example \ref{example-f}.

\noindent{\bf Laplace based Mechanisms.}
The $\ell_1$-norm sensitivity $S_f=2$. W.l.o.g., assume the released answer $\textbf{z}=f(\textbf{s}_5)$. Then for a Laplace mechanism $\mathcal{A}()$, $\frac{Pr(\mathcal{A}(\textbf{s}_5)=f(\textbf{s}_5))}{Pr(\mathcal{A}(\textbf{s}_3)=f(\textbf{s}_5))}=e^{\frac{3}{2}\epsilon}>e^\epsilon$ \footnote{Let $\tilde{\textbf{n}}\in\mathbb{R}^2$ be $2$ i.i.d Laplace noises with mean $0$ and variance $1$. Then $Lap(S_f/\epsilon)=\frac{S_f}{\epsilon}\tilde{\textbf{n}}$ is the Laplace noises added to the query.
$\frac{Pr(\mathcal{A}(\textbf{s}_5)=f(\textbf{s}_5))}{Pr(\mathcal{A}(\textbf{s}_3)=f(\textbf{s}_5))}
=\frac{Pr(f(\textbf{s}_5)+Lap(2/\epsilon)=f(\textbf{s}_5))}{Pr(f(\textbf{s}_3)+Lap(2/\epsilon)=f(\textbf{s}_5))}
=\frac{Pr(\tilde{\textbf{n}}=\frac{\epsilon}{2}[0,0]^T)}{Pr(\tilde{\textbf{n}}=\frac{\epsilon}{2}[1,2]^T)}
=exp(\frac{\epsilon}{2} ({||[1,2]^T||_1-||[0,0]^T||_1}))
=exp(\frac{3}{2}\epsilon)$.
}, $\frac{Pr(\mathcal{A}(\textbf{s}_5)=f(\textbf{s}_5))}{Pr(\mathcal{A}(\textbf{s}_2)=f(\textbf{s}_5))}=e^{\frac{3}{2}\epsilon}>e^\epsilon$ .
Hence $\textbf{s}_5$ is distinct from $\textbf{s}_2$ and $\textbf{s}_3$.
Consequently, if $f(\textbf{s}_5)$ is far from $f(\textbf{s}_2)$ and $f(\textbf{s}_3)$, then $\textbf{s}_5$ will be exposed. 

\noindent{\bf $K$-norm based Mechanisms.}
Sensitivity hull of the query $f$ is $Conv(f(\textbf{s}_2)-f(\textbf{s}_3), f(\textbf{s}_3)-f(\textbf{s}_2))$ where $Conv()$ is the function of deriving convex hull.
For a $K$-norm based mechanism $\mathcal{A}()$, it means $\mathcal{A}(\textbf{s}_2)$ and $\mathcal{A}(\textbf{s}_3)$ are on the line of $\overline{f(\textbf{s}_2)f(\textbf{s}_3)}$ (dashed line through $f(\textbf{s}_2)$ and $f(\textbf{s}_3)$ in Figure \ref{Figure-exposure} (right)); $\mathcal{A}(\textbf{s}_5)$ is on the dashed line through $f(\textbf{s}_5)$ in Figure \ref{Figure-exposure} (right). Again we assume the released result $\textbf{z}=f(\textbf{s}_5)$. Then $Pr(\mathcal{A}(\textbf{s}_3)=f(\textbf{s}_5))=0$. Clearly, any $\textbf{z}$ not on the line of $\overline{f(\textbf{s}_2)f(\textbf{s}_3)}$ leads to complete exposure of $\textbf{s}_5$.
\end{example}

Intuitively, given a disconnected node $\textbf{s}_i$ in the constrained set $\mathcal{C}_t$, if there exists another node $\textbf{s}_j\in\mathcal{C}_t$ such that the difference $f(\textbf{s}_i)-f(\textbf{s}_j)$ is contained in the sensitivity hull $K$, then $\textbf{s}_i$ is protected by $\textbf{s}_j$. Otherwise, if no such node $\textbf{s}_j$ exists, then $\textbf{s}_i$ is exposed. Therefore,
privacy risk can be measured by sensitivity hull $K$ as follows.
If there is no $\textbf{s}_j$ such that $f(\textbf{s}_j)\in f(\textbf{s}_i)+K$, then $\textbf{s}_i$ is exposed, meaning that Equation (\ref{eqn-DPMM}) will not hold.
To capture such geometric meaning (i.e., $f(\textbf{s}_j)-f(\textbf{s}_i)\in K$),
we define degree of protection.
\begin{definition}[\textsc{DoP}]
At any timestamp $t$, the degree of protection (\textsc{DoP}) of a state $\textbf{s}_i$ is the number of states contained in $f(\textbf{s}_i)+K_t$ where $K_t$ is the sensitivity hull.
\begin{align*}
\textsc{DoP}(\textbf{s}_i,K_t)=\left| \{\textbf{s}_j|f(\textbf{s}_j)\in f(\textbf{s}_i)+K_t, \textbf{s}_j\in\mathcal{C}_t\} \right|
\end{align*}
\end{definition}
Because $f(\textbf{s}_i)$ is always in $f(\textbf{s}_i)+K$, $\textsc{DoP}(\textbf{s}_i,K)\geq 1$ for all $\textbf{s}_i\in\mathcal{C}_t$.
Note that not all disconnected nodes are exposed.
For example, in Figure \ref{Figure-Gt-2}, $\textbf{s}_2$ is disconnected under constraint $\mathcal{C}_t=\{\textbf{s}_2,\textbf{s}_4,\textbf{s}_5,\textbf{s}_6\}$. However, $\textsc{DoP}(\textbf{s}_2)=3$ since $f(\textbf{s}_2)+K$ contains $f(\textbf{s}_4)$ and $f(\textbf{s}_5)$.

If all the nodes of a graph have $\textsc{DoP}>1$, we say it is protectable.
 Note that a complete graph is always protectable because every two nodes are connected.

\begin{definition}[Protectable Graph]
\label{def-protectable-G}
A graph $\mathcal{G}$ is protectable if all its nodes have $\textsc{DoP}>1$.
\end{definition}

\begin{theorem}[Exposure Condition]
With Laplace mechanism or $K$-norm based mechanisms, a graph $\mathcal{G}$ cannot satisfy the DPHMM condition (Definition \ref{def-DPMC}) iff $\mathcal{G}$ is not protectable.
\end{theorem}

\noindent{\bf Computation.}
The computation of protectability (i.e. $\textsc{DoP}$) is
to check the number of $f(\textbf{s}_j)$ inside a convex body $f(\textbf{s}_i)+K$ for all $\textbf{s}_j\in\mathcal{C}_t$.
Because the problem of checking whether a point is a convex body
 has been well studied in computational geometry, we skip the discussion of details.

\subsection{Blowfish Analysis}
\label{sec-Blowfish5}
We now use the technique of sensitivity hull to analyze the protection of Blowfish privacy.
\begin{example}[Information Exposure]
\label{example-Blowfish1}
Given the table $\mathcal{T}$ in Figure \ref{Figure-example-TBlowfish} and the graph in Figure \ref{Figure-example-GBlowfish}, let $f$ be a two-dimensional query:
\begin{flalign*}
\hspace{0.5cm}
&f_1:\ \textrm{select count(*) from } \mathcal{T} \textrm{ where disease=``cancer''}&\\
&f_2:\ \textrm{select count(*) from } \mathcal{T} \textrm{ where disease=``diabetes''}&
\end{flalign*}
The $\ell_1$-norm sensitivity $S_f=2$.
By Definition \ref{def-shull1},
$
\Delta f=\left[
\begin{array}{cccc}
1&-1&-1&1\\
0&1&0&-1
\end{array}
\right]^T$.
Then the sensitivity hull can be derived.
Assume $\{\epsilon,G\}$-Blowfish privacy is preserved where $G$ is the graph in Figure \ref{Figure-example-GBlowfish}.
Then secret $\textbf{s}_5$ is protected with both Laplace mechanism and $K$-norm mechanism; while secret $\textbf{s}_6$ is protected only by Laplace mechanism, not by $K$-norm mechanism.
It can be proven that  with $K$-norm mechanism the unbounded differential privacies for the existences of $\{\textbf{s}_1, \textbf{s}_2, \textbf{s}_3, \textbf{s}_4, \textbf{s}_5, \textbf{s}_6\}$ are $\{\epsilon,0,\epsilon,2\epsilon,\epsilon,2\epsilon\}$ respectively
(e.g.\\ $\frac{Pr(\mathcal{A}(D\cup \textbf{s}_6)=
\textbf{z})}{Pr(\mathcal{A}(D)=\textbf{z})}\leq e^{2\epsilon}$).
 Thus it is $2\epsilon$-unbounded-DP in total.
 This also illustrates why we need to re-design a new optimal graph with less privacy loss in Section \ref{sec-complexity-OG}. Note that
 (1) the original bounded Blowfish privacy (i.e. the graph in Figure \ref{Figure-example-GBlowfish}) does not protect $\textbf{s}_6$ in the first place. Hence the original Blowfish privacy still holds;
 (2) although $\textbf{s}_4$ is connected with $\textbf{s}_3$, the existence of $\textbf{s}_4$ is also at risk ($2\epsilon$-DP);
 (3) although $\textbf{s}_6$ is not connected with $\textbf{s}_4$, it is indistinguishable with $\textbf{s}_4$ by default.
\end{example}
Formally, we summarize the protection of Blowfish as follows.
(1)
Bounded Blowfish and unbounded Blowfish interfere with each other, e.g., bounded Blowfish can ensure or violate unbounded Blowfish privacy and vice versa.
(2)
For various queries, the protection of Blowfish differs.
(3) For various data release mechanisms, the protection of Blowfish differs. 

\vspace{2mm}
\noindent{\bf Quantifying Blowfish.}
We quantify the overall protection of Blowfish as follows. First, it is intuitive that bounded Blowfish is weaker than (or equal to) bounded DP, and unbounded Blowfish is also weaker than (or equal to) unbounded DP. For lack of space, below we quantify the protection of bounded Blowfish in terms of unbounded DP with $K$-norm mechanism. It can be easily extended to other cases.
%
%
\begin{definition}[$\{\epsilon,G,\mathcal{C}\}$-$Constrained$DP]
\label{def-CDP}
Let $G=(\mathcal{S},\mathcal{E})$ be the policy graph in Blowfish privacy,
and $\mathcal{C}$ be the instances satisfying the constraint in either Markov model or database context.
A randomized mechanism $\mathcal{A}()$ satisfies $\{\epsilon,G,\mathcal{C}\}$-constrained differential privacy if for any output $\textbf{z}$, one of the following condition holds:
\begin{flalign*}
\begin{array}{l}
\textrm{(1). in Markov model, } \frac{Pr(\mathcal{A}(\textbf{s}_j)=\textbf{z})}{Pr(\mathcal{A}(\textbf{s}_k)=\textbf{z})}\leq e^\epsilon, \forall \textbf{s}_j,\textbf{s}_k\in\mathcal{C};
\\
\textrm{(2). in database context, } \frac{Pr(\mathcal{A}(D_1)=\textbf{z})}{Pr(\mathcal{A}(D_2)=\textbf{z})}\leq e^\epsilon, \forall D_1,D_2,\in\mathcal{C}, D_1 \\\textrm{can be obtained by adding $\textbf{s}_i$ to or removing $\textbf{s}_i$ from } D_2, \textbf{s}_i\in\mathcal{S}.
\end{array}
\end{flalign*}
\end{definition}
Note that in Markov model, above definition is actually bounded DP because unbound DP becomes bounded DP by nature (Section \ref{sec-Privacy-constraint}).


\begin{figure*}
\centering
\begin{subfigure}{0.192\textwidth}
\centering
\includegraphics[width=3.4cm]{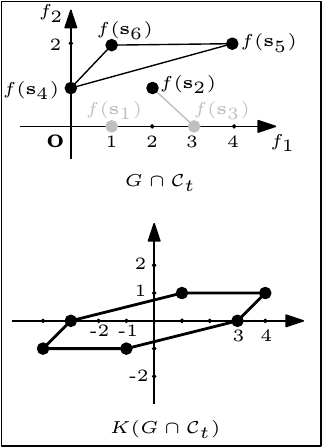}
\caption{{\small $\mathcal{C}_t=\{\textbf{s}_2,\textbf{s}_4,\textbf{s}_5,\textbf{s}_6\}$}}
\label{Figure-Gt-2}
\end{subfigure}
\begin{subfigure}{0.192\textwidth}
\centering
\includegraphics[width=3.4cm]{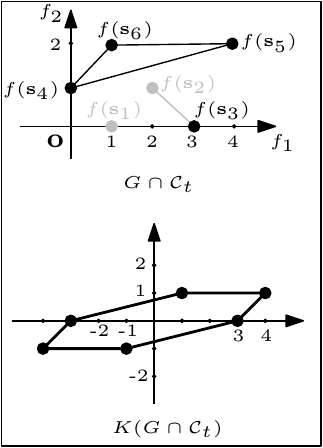}
\caption{{\small $\mathcal{C}_t=\{\textbf{s}_3,\textbf{s}_4,\textbf{s}_5,\textbf{s}_6\}$}}
\label{Figure-Gt-3}
\end{subfigure}
\begin{subfigure}{0.192\textwidth}
\centering
\includegraphics[width=3.4cm]{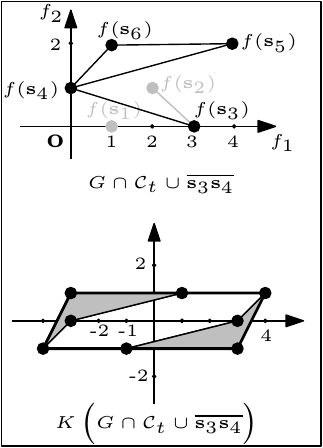}
\caption{{\small $\mathcal{C}_t=\{\textbf{s}_3,\textbf{s}_4,\textbf{s}_5,\textbf{s}_6\}$}}
\label{Figure-Gt-4}
\end{subfigure}
\begin{subfigure}{0.192\textwidth}
\centering
\includegraphics[width=3.4cm]{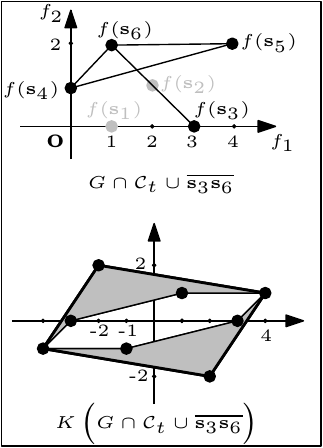}
\caption{{\small $\mathcal{C}_t=\{\textbf{s}_3,\textbf{s}_4,\textbf{s}_5,\textbf{s}_6\}$}}
\label{Figure-Gt-5}
\end{subfigure}
\begin{subfigure}{0.192\textwidth}
\centering
\includegraphics[width=3.4cm]{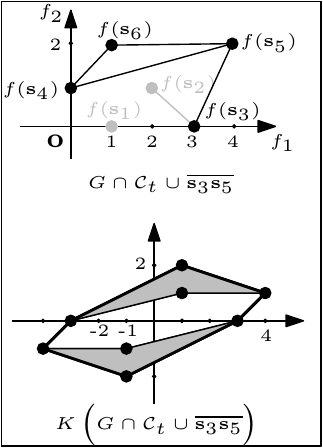}
\caption{{\small $\mathcal{C}_t=\{\textbf{s}_3,\textbf{s}_4,\textbf{s}_5,\textbf{s}_6\}$}}
\label{Figure-Gt-6}
\end{subfigure}
\caption{{\small (a): if $\mathcal{C}_t=\{\textbf{s}_2,\textbf{s}_4,\textbf{s}_5,\textbf{s}_6\}$, then $\textbf{s}_2$ is also protected because $f(\textbf{s}_4)\in f(\textbf{s}_2)+K$ and $f(\textbf{s}_5)\in f(\textbf{s}_2)+K$; (b): if $\mathcal{C}_t=\{\textbf{s}_3,\textbf{s}_4,\textbf{s}_5,\textbf{s}_6\}$, then $\textbf{s}_3$ is exposed; (c): adding $\overline{\textbf{s}_3\textbf{s}_4}$ to graph; (d): adding $\overline{\textbf{s}_3\textbf{s}_6}$ to graph; (e): adding $\overline{\textbf{s}_3\textbf{s}_5}$ to graph;  }}
\label{Figure-Gt-example}
\end{figure*}

\begin{theorem}[Blowfish Protection]
\label{thm-Blowfish-DP}
Let $G$ be the policy graph, and $\mathcal{C}$ be the instances satisfying the constraint in Blowfish privacy.
With $K$-norm mechanism, if $\{\epsilon, \{\mathcal{S},G,\mathcal{C}\}\}$-Blowfish privacy holds, then it satisfies \\
(1). $\left\{\left(\mathop{max}\limits_{\forall \textbf{s}_j,\textbf{s}_k\in\mathcal{C}}||f(\textbf{s}_j)-f(\textbf{s}_k)||_{K}\right)\epsilon,G,\mathcal{C}\right\}\textrm{-}constrained$DP in Markov model;\\
(2). $\left\{\left(\mathop{max}\limits_{\forall \textbf{s}_i\in\mathcal{C}}||f(\textbf{s}_i)||_{K}\right)\epsilon,G,\mathcal{C}\right\}\textrm{-}constrained$DP in database context,\\
where $K$ is the sensitivity hull of query $f$.

%
\end{theorem}

\section{Data Release Mechanism}
If privacy risk is detected, we build a protectable graph as a supergraph of existing graph \footnote{On the other hand, since the policy graph is customizable to users, the protectable graph can also be created by users. In this case, it is not necessary to derive another minimum protectable graph again.}. In this section, we first formulate the problem of building a minimum protectable graph with lowest error bound. Next we show that this problem is $\#$P-hard, and propose a fast greedy algorithm. Then we present the data release mechanism.

\subsection{Minimum Protectable Graph}
\label{sec-complexity-OG}
It is clear that a protectable graph satisfies the DPHMM condition in Definition \ref{def-DPMC}. Therefore, when information is exposed, we need to build a protectable graph by re-connecting the disconnected nodes so that they have $\textsc{DoP}>1$.
 Next we
 formulate the problem of building a minimum protectable graph and investigate its computational complexity,
 then propose a greedy algorithm to this end.
%

\vspace{2mm}
\noindent{\bf Minimum Protectable Graph.}
Because the error bound of differential privacy is determined by the volume of sensitivity hull $K(\mathcal{G}_t)$ \cite{Geometry-Hardt-STOC10} where $\mathcal{G}_t$ is the graph under constraint at timestamp $t$, the optimal graph should have the minimum volume of $K(\mathcal{G}_t)$ for best utility. We define the optimal graph as follows.

Given the policy graph $G$ and the constraint set $\mathcal{C}_t$ at timestamp $t$, the optimal graph $\widehat{G}_t$ is a graph containing $G\cap\mathcal{C}_t$ with minimum volume $K(\widehat{G}_t)$ under the DPHMM condition (Definition \ref{def-DPMC}):
\begin{align}
\label{eqn-opt-G}
\widehat{G}_t=\mathop{argmin}\limits_{\mathcal{G}}\textsc{Vol}(K(\mathcal{G}))
\end{align}
\vspace{-0.5cm}
\begin{flalign*}
\hspace{1.2cm}
\textrm{ subject to: } &(G\cap\mathcal{C}_t)\subseteq\widehat{G}_t& \\
&\widehat{G}_t \textrm{ satisfies the DPHMM condition}&
\end{flalign*}

\begin{example}[Minimum Protectable Graph]
\label{example-mpg}
Given the query in Example \ref{example-f} and the graph in Figure \ref{Figure-example-03},
Figure \ref{Figure-Gt-3} shows the graph under constraint $\mathcal{C}_t=\{\textbf{s}_3,\textbf{s}_4,\textbf{s}_5,\textbf{s}_6\}$. Then $\textbf{s}_3$ is exposed because $f(\textbf{s}_3)+K$ contains no other node. To satisfy the DPHMM condition, we need to connect $\textbf{s}_3$ to
another node in $\mathcal{C}_t$, i.e. $\textbf{s}_4$, $\textbf{s}_5$ or $\textbf{s}_6$.

If $\textbf{s}_3$ is connected to $\textbf{s}_4$, then Figure \ref{Figure-Gt-4} shows the new graph and its sensitivity hull. By adding two new edges $\{f(\textbf{s}_3)-f(\textbf{s}_4),f(\textbf{s}_4)-f(\textbf{s}_3)\}$ to $\Delta f$, the shaded areas are attached to the sensitivity hull. Similarly, Figures \ref{Figure-Gt-5} and \ref{Figure-Gt-6} show the new sensitivity hulls when $\textbf{s}_3$ is connected to $\textbf{s}_6$ and $\textbf{s}_5$ respectively. Because the smallest $\textsc{Area}(K)$ is in Figure \ref{Figure-Gt-4}, the optimal graph $\widehat{G}_t$ is $G\cap\mathcal{C}_t\cup\overline{\textbf{s}_3\textbf{s}_4}$.
\end{example}

\noindent{\bf Complexity.}
 We can see that to derive the optimal graph, minimum volume $\textsc{Vol}(K)$ should be computed.
 For any query $f:\mathcal{S}\rightarrow\mathbb{R}^d$, $K$ is a polytope in $\mathbb{R}^d$.
 However, the volume computation of polytope is $\#$P-hard \cite{Dyer:1988:CCV:65086.65094}. Thus it follows that the computation of minimum volume is no easier than $\#$P-hard
 \footnote{In low-dimensional space, it is still possible to design fast algorithms. For example, minimum protectable graph can be derived in $O(nm^3)$ time in $2$-dimensional space where $m=|\mathcal{E}|$ is the number of edges and $n$ is the number of exposed nodes.}.
\begin{theorem}
The problem of minimum protectable graph in Equation (\ref{eqn-opt-G}) is $\#$P-hard.
\end{theorem}

\noindent{\bf Greedy Algorithm.}
Due to the computational complexity, we propose a greedy algorithm similar to minimum spanning tree. The idea is to connect each disconnected node to its nearest (in measurement space) node.
For other theoretical algorithms of volume computation with polynomial time bound, please see \cite{Vempala05geometricrandom}.
Algorithm \ref{alg-Gt-connected} shows the greedy algorithm, which takes
$O(N^2)$ time where $N=|\mathcal{V}|$ is the number of nodes.
\begin{algorithm}[htb]
\caption{Protectable Graph}
\begin{algorithmic}[1]
\Require{
$G$, $\mathcal{C}_t$, $f$
}
\State{$\mathcal{G}_t\gets G\cap\mathcal{C}_t$;}
\ForAll{exposed node $\textbf{s}_i\in\mathcal{C}_t$}
\State{$\textbf{s}_j\gets \mathop{argmin}\limits_{\textbf{s}\in\mathcal{C}_t}||f(\textbf{s})-f(\textbf{s}_i)||_2$;}
\State{$\mathcal{G}_t\gets\mathcal{G}_t\cup\overline{\textbf{s}_i\textbf{s}_j}$;}
\Comment{{\tt \scriptsize connect to nearest node}}
\EndFor
\\\Return{protectable graph $\mathcal{G}_t$;}
\end{algorithmic}
\label{alg-Gt-connected}
\end{algorithm}

\subsection{Data Release Mechanism}
\label{subsec-framework}

The data release mechanism is shown in Algorithm \ref{alg-framework}.
At each timestamp $t$, we compute the prior probability vector $\textbf{p}_t^-$.
Under the constraint $\mathcal{C}_t$,
the graph $G$ becomes a subgraph $G \cap \mathcal{C}_t$.
To satisfy the DPHMM condition, we derive a protectable graph $\mathcal{G}_t$ by Algorithm \ref{alg-Gt-connected}.
Next a differentially private mechanism can be adopted to release a perturbed answer $\textbf{z}_t$.
Then the released $\textbf{z}_t$ will also be used to update the posterior probability $\textbf{p}_t^+$ (in the equation below) by Equation (\ref{eqn-posterior}), which subsequently will be used to compute the prior probability for the next timestamp $t+1$.
\begin{align*}
\textbf{p}_t^+[i]=Pr(\textbf{s}_t^*=\textbf{s}_i|\textbf{z}_t,\textbf{z}_{t-1},\cdots,\textbf{z}_1)
\end{align*}

%

\begin{algorithm}[htb]
\caption{Data Release Mechanism}
\begin{algorithmic}[1]
\Require{
$\epsilon_t$, $G$, $f$, $\textbf{M}$, $\textbf{p}_{t-1}^+$, $\textbf{s}_t^*$
}
\State{$\textbf{p}_{t}^-\gets\textbf{p}_{t-1}^+\textbf{M}$;}
\Comment{{\tt \scriptsize Markov transition}}
\State{$\mathcal{C}_t\gets\{\textbf{s}_i|\textbf{p}_t^-[i]>0\}$;}
\Comment{{\tt \scriptsize constraint}}
\State{$\mathcal{G}_t \gets$ \Call{Algorithm \ref{alg-Gt-connected}}{$G$, $\mathcal{C}_t$, $f$};}
\Comment{{\tt \scriptsize protectable graph $\mathcal{G}_t$}}
\State{$\textbf{z}_t\gets$ $K$-norm based mechanism($f(\textbf{s}_t^*)$, $K(\mathcal{G}_t)$);}
\label{line-releasing}
\State{Derive $\textbf{p}_t^+$ by Equation (\ref{eqn-posterior});}
\Comment{{\tt \scriptsize inference}}
\label{line-bayesian}
\LineComment{{\tt \scriptsize go to next timestamp}}
\\\Return{\Call{Algorithm \ref{alg-framework}}{$\epsilon_{t+1}$, $G$, $f$, $\textbf{M}$, $\textbf{p}_{t}^+$, $\textbf{s}_{t+1}^*$}};
\end{algorithmic}
\label{alg-framework}
\end{algorithm}
Note that in line \ref{line-releasing} $\textbf{z}_t$ can be released by 
either Laplace mechanism or $K$-norm mechanism. For simplicity, we use $K$-norm mechanism as a unifying mechanism (Theorem \ref{theo-lap-K}).

\begin{theorem}
Given policy graph $G$ and query $f$,
Algorithm \ref{alg-framework} satisfies $\{\epsilon,G,\mathcal{C}_t\}$-DPHMM
at any timestamp.
\end{theorem}
%
\begin{theorem}
Given policy graph $G$ and query $f$,
at any timestamp $t$,
Algorithm \ref{alg-framework} satisfies\\ $\left\{\left(\mathop{max}\limits_{\forall \textbf{s}_j,\textbf{s}_k\in\mathcal{C}_t}||f(\textbf{s}_j)-f(\textbf{s}_k)||_{K_t}\right)\epsilon,G,\mathcal{C}_t\right\}\textrm{-}constrained$DP (Definition \ref{def-CDP}) where $\mathcal{C}_t$ is the constraint, $K_t$ is the sensitivity hull of query $f$ and protectable graph $\mathcal{G}_t$.
\end{theorem}

\section{Privacy Composition}
In some cases, multiple queries need to be answered.
Thus we analyze the privacy composition for multiple data releases. Note that the parallel composition \cite{McSherry-PINQ} is not applicable because there is only one state in Markov model.

\vspace{2mm}
\noindent{\bf Single-Time Multiple-Queries. }
At one timestamp, it is possible that many queries should be answered. Then the privacy cost $\epsilon$ composes for all queries.
\begin{theorem}
At timestamp $t$, an $\{\epsilon,G,\mathcal{C}_t\}$-DPHMM mechanism released multiple answers $\textbf{z}_1,\textbf{z}_2,\cdots,\textbf{z}_n$ for queries $f_1,f_2,\cdots,f_n$ with $\epsilon_1,\epsilon_2,\cdots,\epsilon_n$, then it satisfies $\{\sum_{i=1}^n \epsilon_i,G,\mathcal{C}_t\}$-DPHMM and
$\left\{\mathop\sum\limits_{i=1}^n\left(\mathop{max}\limits_{\forall \textbf{s}_j,\textbf{s}_k\in\mathcal{C}_t}||f(\textbf{s}_j)-f(\textbf{s}_k)||_{K_i}\right)\epsilon_i,G,\mathcal{C}_t\right\}$-\\$constrained$DP
where $K_i$ denotes the sensitivity hull of $f_i$.
\end{theorem}

\noindent{\bf Multiple-Time Single-Query. }
If a query was answered over multiple timestamps, then the privacy protection has to be enforced on the sequence.
Under the probabilistic constraint,
 we define differentially private sequence with all possible sequences.
\begin{definition}
A constraint set of sequences
 $\mathcal{Q}=\{\textbf{Q}_1,\\\textbf{Q}_2,\cdots,\textbf{Q}_n\}$ is a set of $n$ possible sequences with $Pr(\textbf{Q}_i)>0$ for all $ \textbf{Q}_i\in\mathcal{S}^t$, $i=1,2,\cdots,n$.
\end{definition}

\begin{definition}[$\{\epsilon,\mathcal{Q}\}$-$Constrained$DPS]
During timestamps $1,2,\cdots,t$ in an HMM,  a randomized mechanism $\mathcal{A}()$ generates $\{\epsilon,\mathcal{Q}\}$-$Constrained$DPS if for any output sequence $\textbf{z}_1,\textbf{z}_2,\cdots,\textbf{z}_t$ and any possible sequences $\textbf{Q}_j$ and $\textbf{Q}_k$ in $\mathcal{Q}$, the following holds
\begin{align*}
\frac{Pr\left(\mathcal{A}(\textbf{Q}_j)=(\textbf{z}_1,\textbf{z}_2,\cdots,\textbf{z}_t)\right)}
{Pr\left(\mathcal{A}(\textbf{Q}_k)=(\textbf{z}_1,\textbf{z}_2,\cdots,\textbf{z}_t)\right)}
\leq e^{\epsilon}
\end{align*}
\end{definition}
\begin{theorem}
During timestamps $i=1,2,\cdots,t$ in an $\{\epsilon_i,G,\mathcal{C}_i\}$-DPHMM with policy graph $G$ and constraints
$\mathcal{C}_i=\{\textbf{Q}_j[i]|\forall \textbf{Q}_j\in\mathcal{Q}\}$, 
the released sequence $\textbf{z}_1,\textbf{z}_2,\cdots,\textbf{z}_t$ for a query $f$ satisfies $\left\{\mathop\sum\limits_{i=1}^t\left(\mathop{max}\limits_{\forall \textbf{s}_j,\textbf{s}_k\in\mathcal{C}_i}||f(\textbf{s}_j)-f(\textbf{s}_k)||_{K_i}\right)\epsilon_i,\mathcal{Q}\right\}$-\\$constrained$DPS where $K_i$ denotes the sensitivity hull at timestamp $i$.
\end{theorem}

Above compositions can be combined for the case of multiple-time and multiple-queries data releases. This completes our analysis of privacy composition over time.

\section{Empirical Evaluation}
We report the experimental evaluation in this section. All algorithms were implemented in Matlab on a PC with 2.4GHz CPU and 4GB memory.

\vspace{2mm}
\noindent{\bf Datasets.} We used 
the following two datasets with similar configurations in \cite{LocPriv14-arXiv} for comparison purpose.
The Markov models were learned from the raw data.
From each dataset, $20$ sequences, each of which contains $100$ timestamps, were selected for our experiment. Then the average result is reported.
\begin{itemize}
\item
Geolife dataset. Geolife dataset \cite{GeoLife-Zheng-10} recorded a wide range of users' outdoor movements, represented by a series of tuples containing latitude, longitude and timestamp.
We extracted all the trajectories within the $3$rd ring of Beijing to learn the Markov model, with the map partitioned into cells of  $0.34\times 0.34\ {km}^2$.
\item
Gowalla dataset. Gowalla dataset \cite{KDD-Gowalla-2011} contains $6,442,890$ check-in locations of $196,586$ users over 20 months. We extracted all the check-ins in Los Angeles to train the Markov model, with the map partitioned into cells of $0.89\times 0.89\ {km}^2$.
\end{itemize}

\noindent{\bf Mechanisms.} 
For better utility, we used the planar isotropic mechanism in \cite{LocPriv14-arXiv} (with $\delta=0.01$) to release the locations of users.
We denote our privacy notion and \cite{LocPriv14-arXiv} by DPMM and DPLS \footnote{Differential privacy on location set.} respectively.
Because Laplace mechanism provides no better utility than $K$-norm based mechanism, proved in Corollary \ref{coro-lap-K}, we skipped the evaluation of Laplace mechanism. The default value of $\epsilon$ is $1$ if not mentioned.

\vspace{2mm}
\noindent{\bf Application.}
For location data, a common application is to release the location coordinates. Thus we use the measurement query $f:\mathcal{S}\rightarrow\mathbb{R}^2$ that returns a $2\times 1$ vector of longitude and latitude.

Two policy graphs were adopted in our experiments: utility-oriented $G_{util}$ and privacy-oriented $G_{trs}$, as defined in Section \ref{sec-policy-graph}.
\begin{itemize}
\item
$G_{util}$ connects all nodes if their distances of locations are less than $r$;
\vspace{-1mm}
\item
$G_{trs}$ guarantees that even if the previous states were completely exposed, privacy can still be protected in the current timestamp.
\end{itemize}
Because of different customizations of the two graphs, we can examine the different results of them. In $G_{util}$, the default values of $r$ for GeoLife and Gowalla are $1(km)$ and $2(km)$ respectively.

\vspace{2mm}
\noindent{\bf Metrics.}
We used the following metrics in our experiment.
\begin{itemize}
\item
To measure the efficiency, the runtime of data release method was evaluated.
\vspace{-1mm}
\item
\textsc{DoP} represents the number of nodes that a node is hidden in. Hence
to reflect the privacy level of true states,
\textsc{DoP} of true states was computed.
\vspace{-1mm}
\item
The utility of DPHMM was measured by $\textsc{Error}=\\||\textbf{z}_t-f(\textbf{s}_t^*)||_2$ where $\textbf{z}_t$ is the released answer and $f(\textbf{s}_t^*)$ is the true answer.
\end{itemize}

\begin{figure}[t]
\centering
\begin{subfigure}{0.233\textwidth}
\centering
\includegraphics[width=4.2cm]{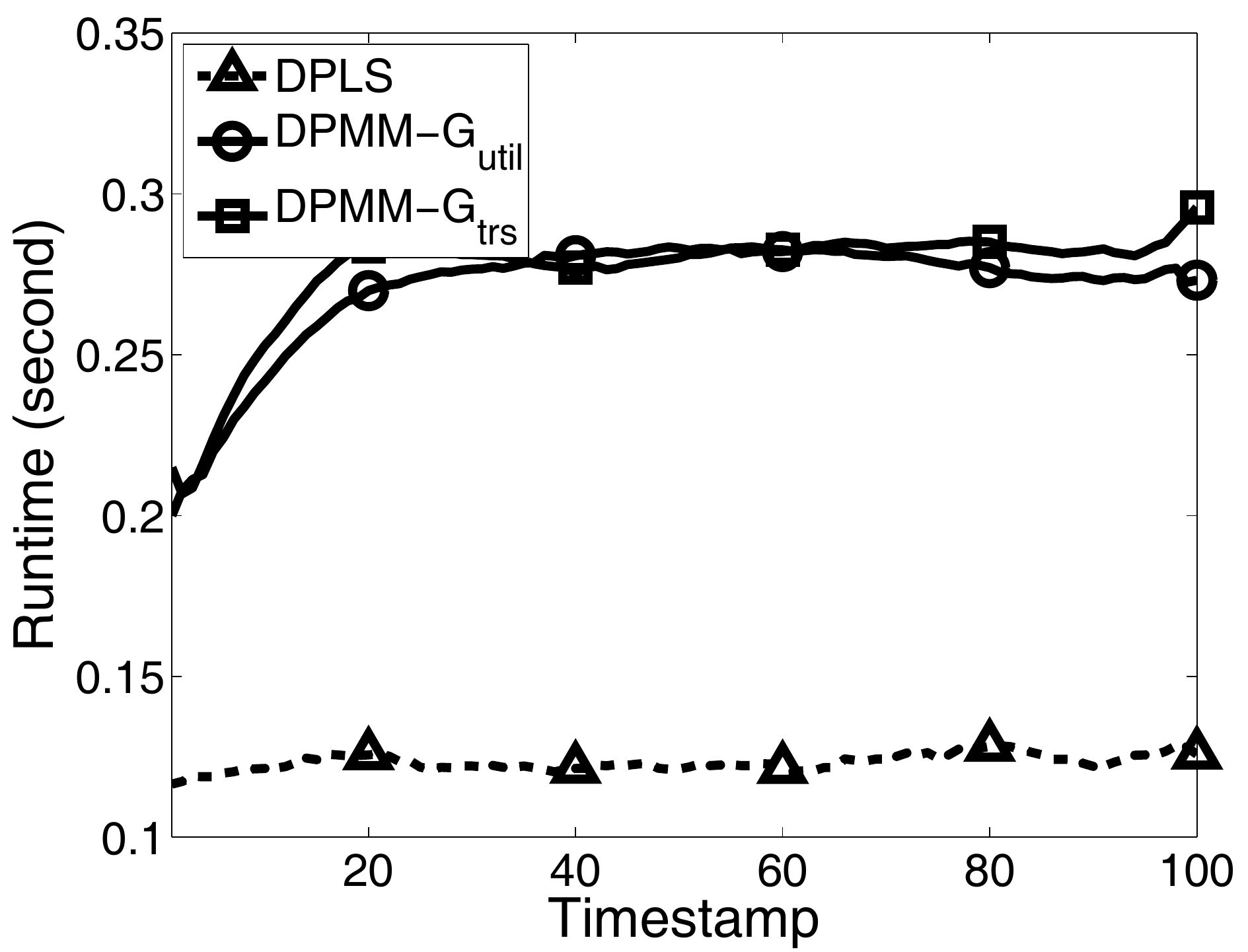}
\caption{{\small Runtime on GeoLife}}
\label{Fig-runtime-GeoLife}
\end{subfigure}
\begin{subfigure}{0.233\textwidth}
\centering
\includegraphics[width=4.2cm]{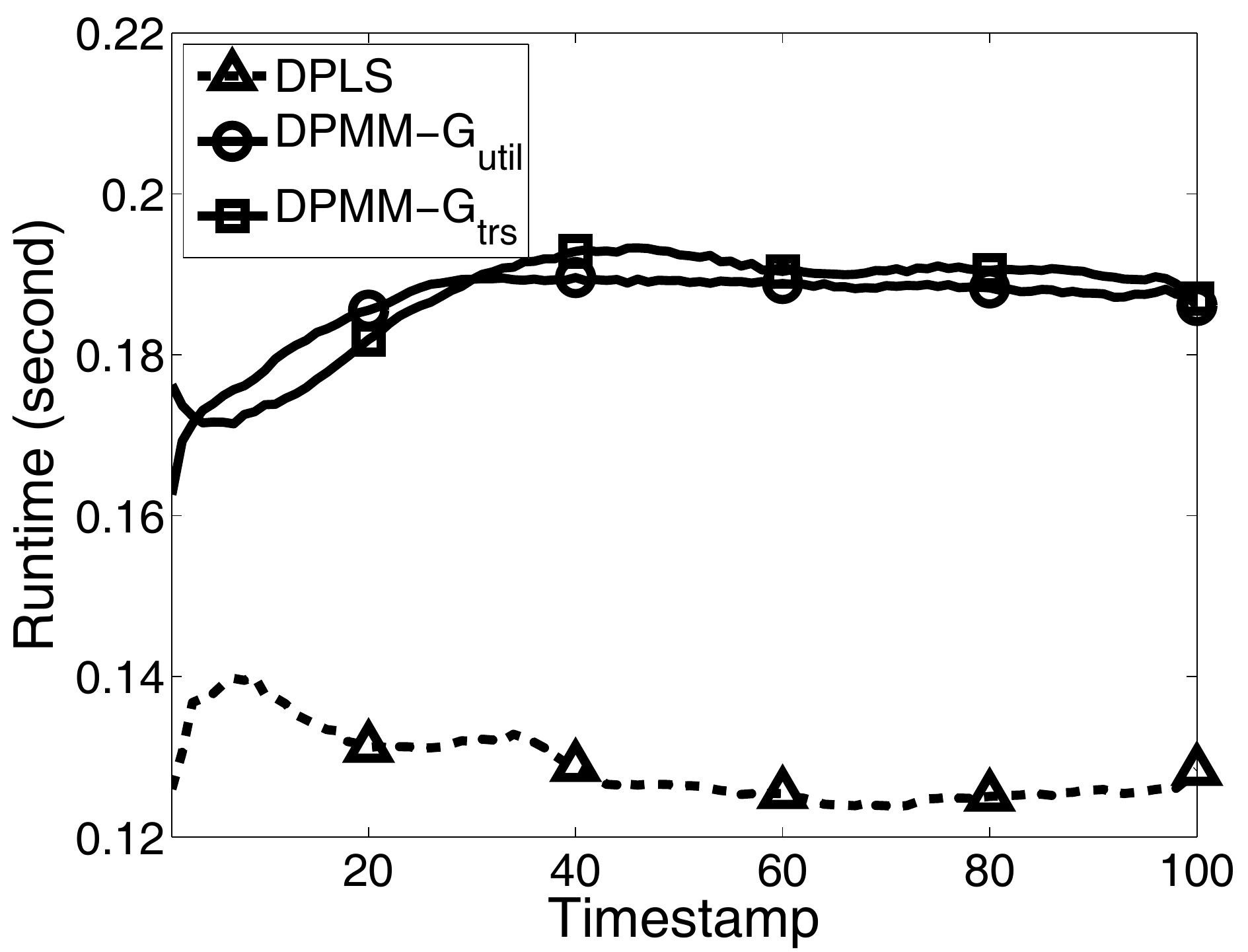}
\caption{{\small Runtime on Gowalla}}
\label{Fig-runtime-Gowalla}
\end{subfigure}
\caption{{\small Runtime.}}
\label{Fig-runtime-all}
\end{figure}
\subsection{Runtime}
Figure \ref{Fig-runtime-all} shows the runtime report on the two datasets. We can see that the runtime of DPHMM, either with $G_{util}$ or $G_{trs}$, is a little bit longer than DPLS. The reason is that
DPLS uses a tighter constraint than $\mathcal{C}_t$, which in our setting became numerous when Markov model converged to a stationary distribution gradually. Then the computation of sensitivity hull took more time with larger graph.
It is also worth noting that sensitivity hull converges with $\mathcal{C}_t$.
As time evolves, the runtime also converges with Markov model to a stable level.

\begin{figure}[!t]
\centering
\begin{subfigure}{0.233\textwidth}
\centering
\includegraphics[width=4.2cm]{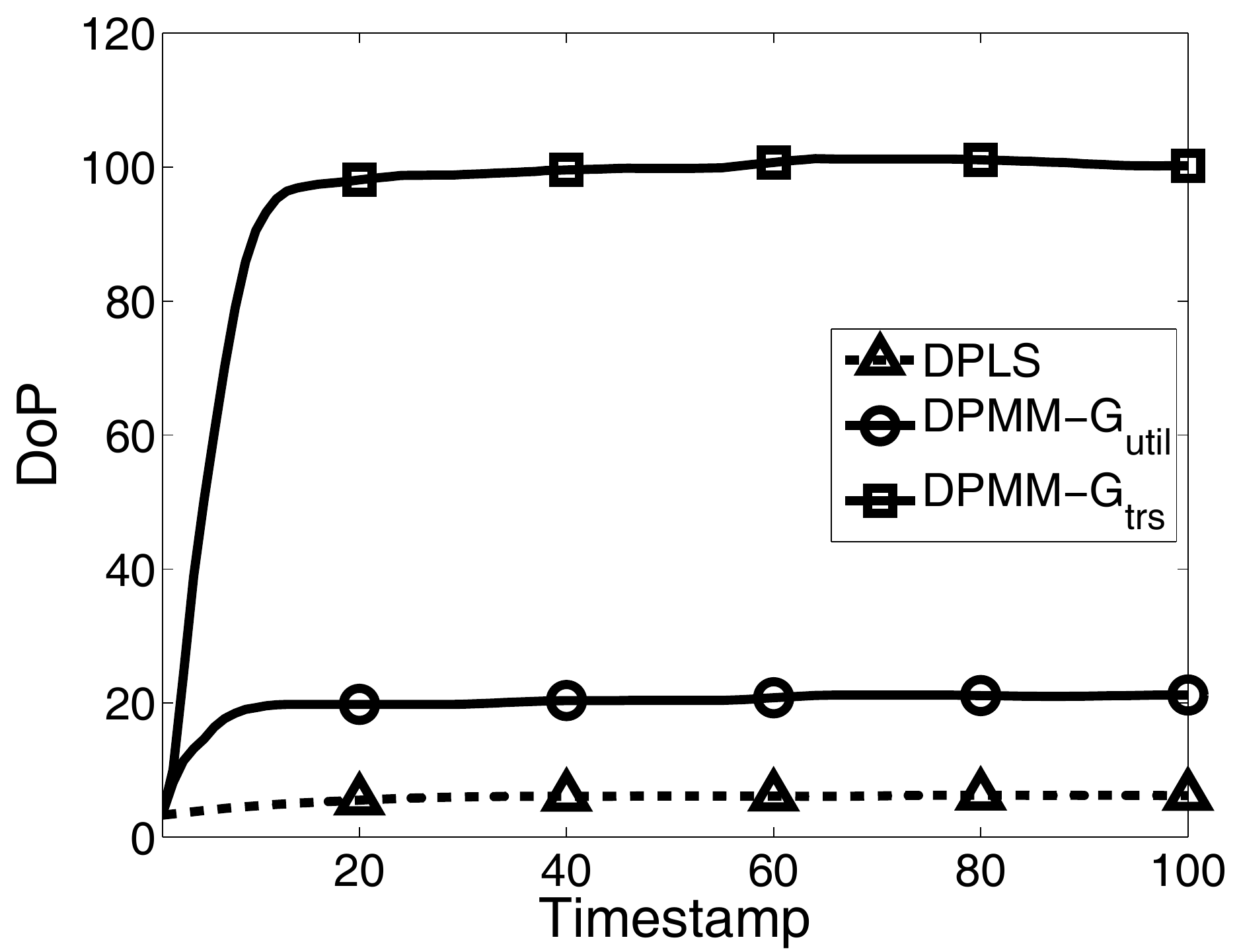}
\caption{{\small \textsc{DoP} on GeoLife}}
\label{Fig-DOP-time-GeoLife}
\end{subfigure}
\begin{subfigure}{0.233\textwidth}
\centering
\includegraphics[width=4.2cm]{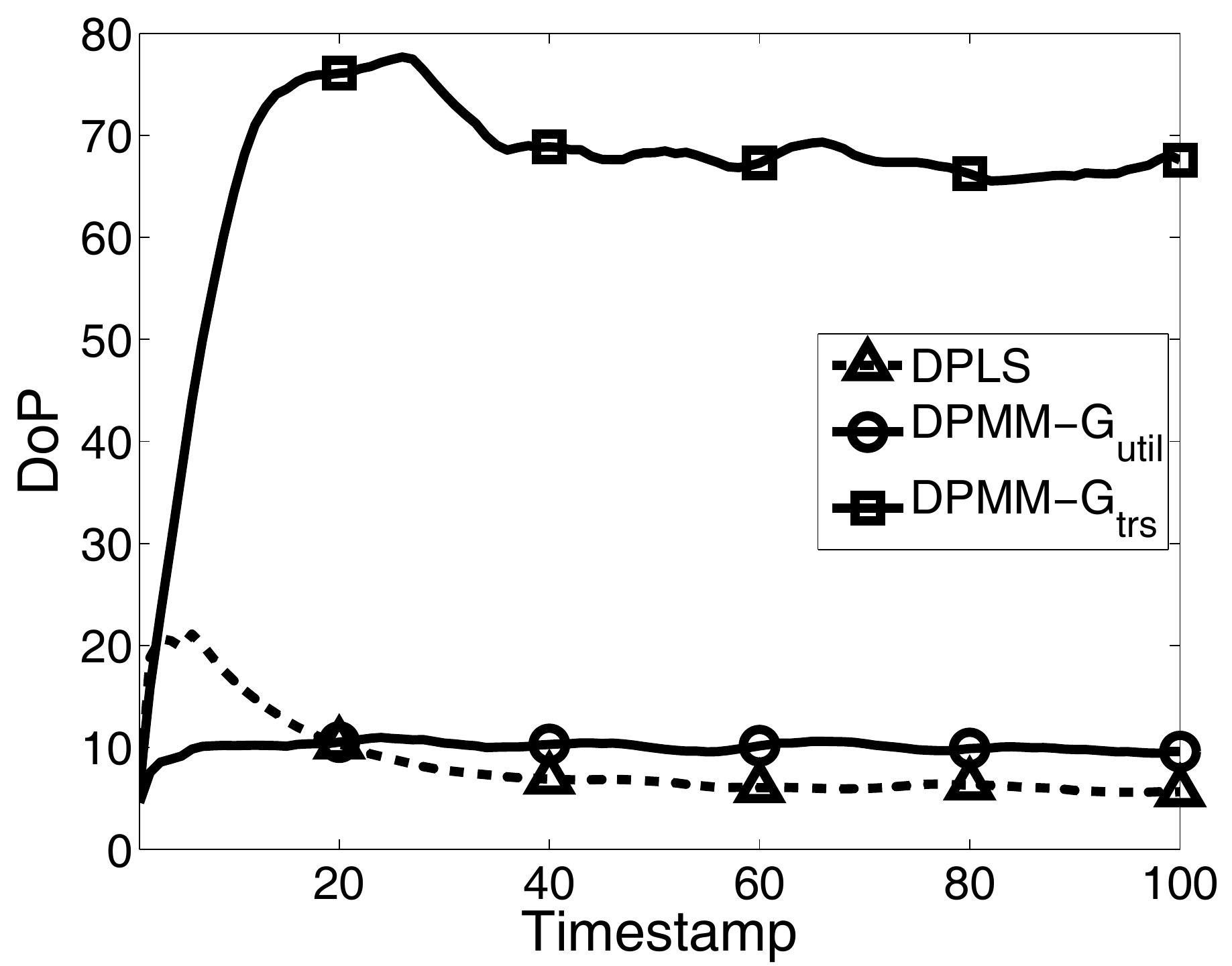}
\caption{{\small \textsc{DoP} on Gowalla}}
\label{Fig-DOP-time-Gowalla}
\end{subfigure}
\\
\begin{subfigure}{0.233\textwidth}
\centering
\includegraphics[width=4.2cm]{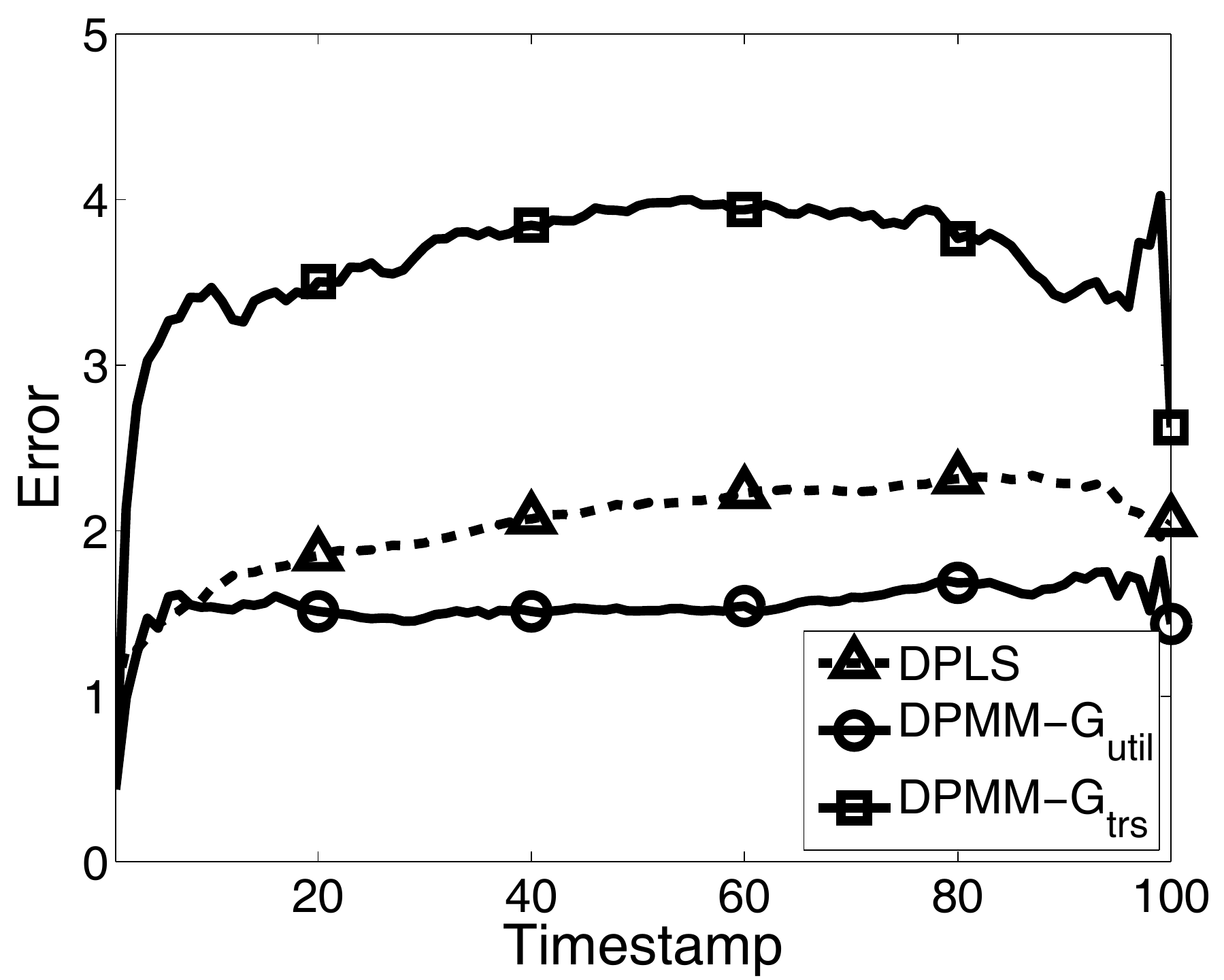}
\caption{{\small \textsc{Error} on GeoLife}}
\label{Fig-Error-time-GeoLife}
\end{subfigure}
\begin{subfigure}{0.233\textwidth}
\centering
\includegraphics[width=4.2cm]{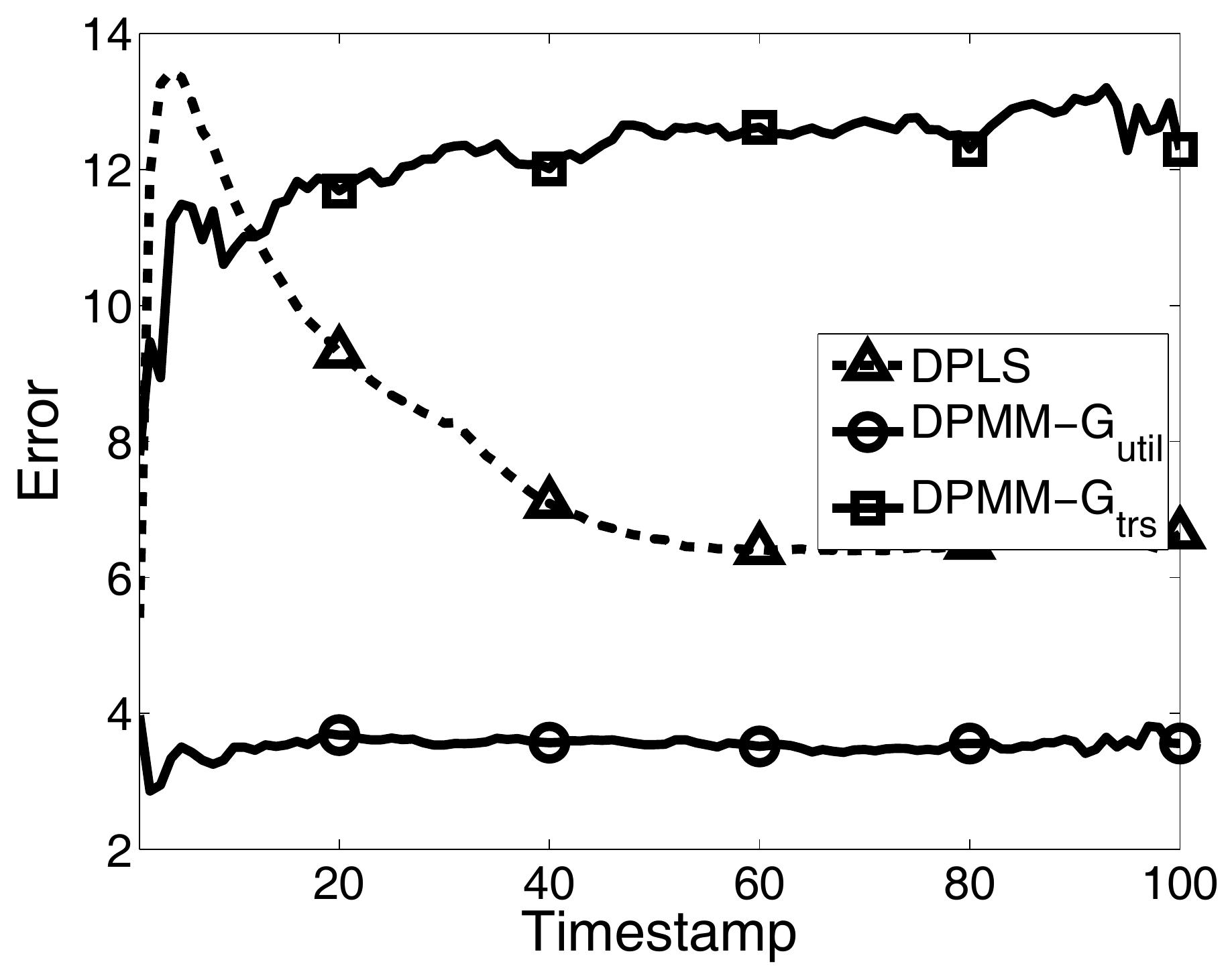}
\caption{{\small \textsc{Error} on Gowalla}}
\label{Fig-Error-time-Gowalla}
\end{subfigure}
\caption{{\small Performance over time.}}
\label{Fig-performance-time}
\end{figure}
\begin{figure}[!t]
\centering
\begin{subfigure}{0.233\textwidth}
\centering
\includegraphics[width=4.2cm]{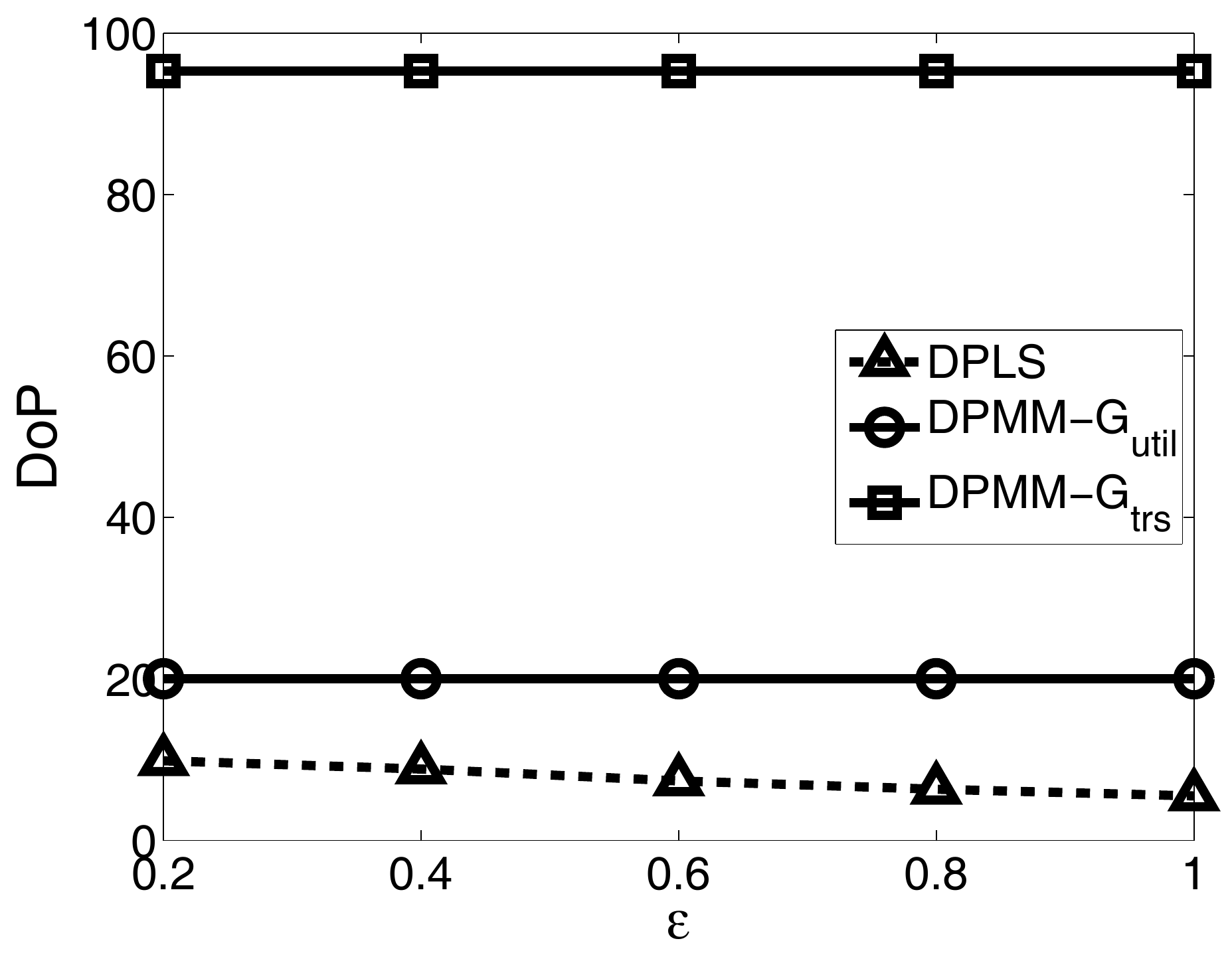}
\caption{{\small \textsc{DoP} on GeoLife}}
\label{Fig-DOP-eps-GeoLife}
\end{subfigure}
\begin{subfigure}{0.233\textwidth}
\centering
\includegraphics[width=4.2cm]{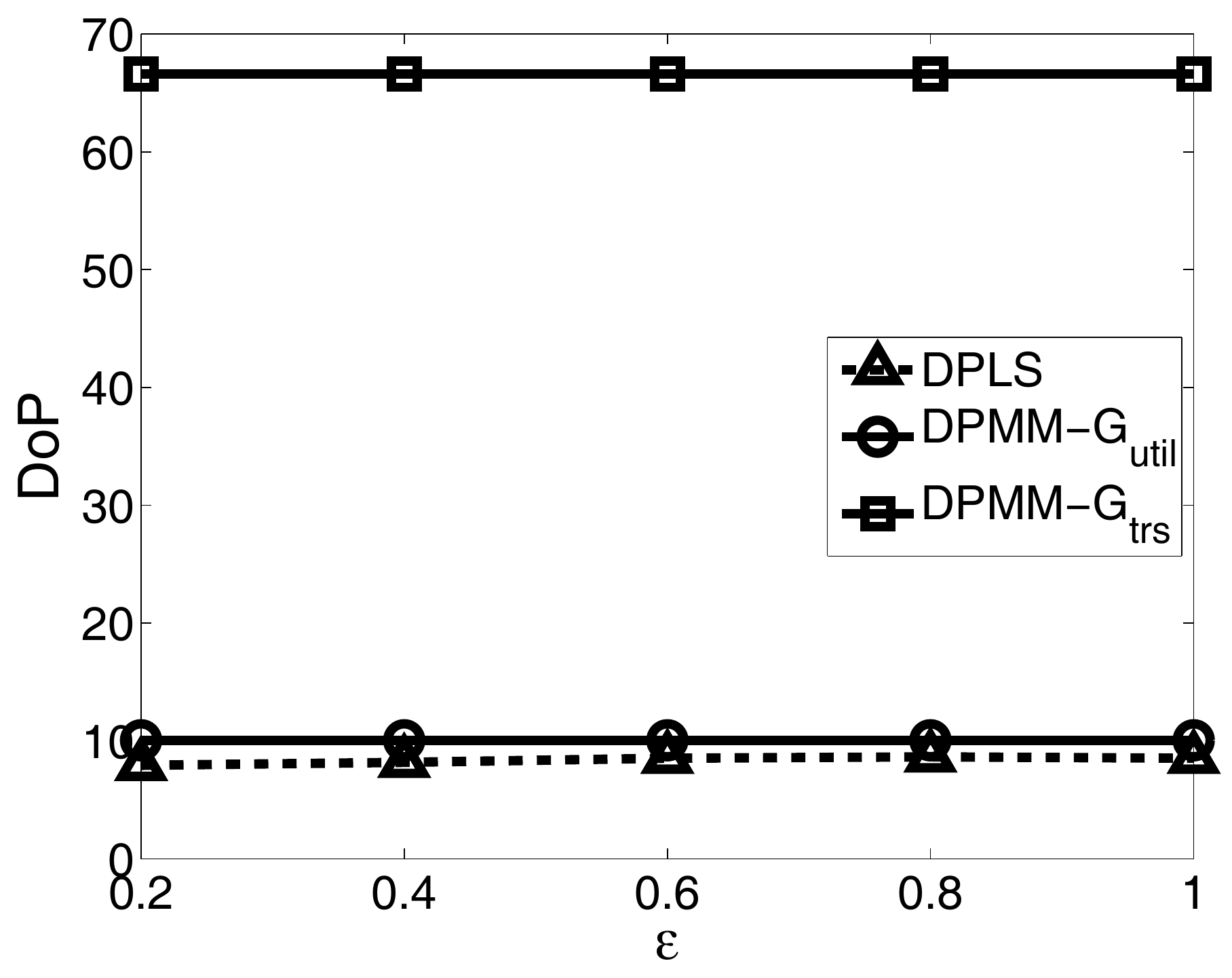}
\caption{{\small \textsc{DoP} on Gowalla}}
\label{Fig-DOP-eps-Gowalla}
\end{subfigure}
\\
\begin{subfigure}{0.233\textwidth}
\centering
\includegraphics[width=4.2cm]{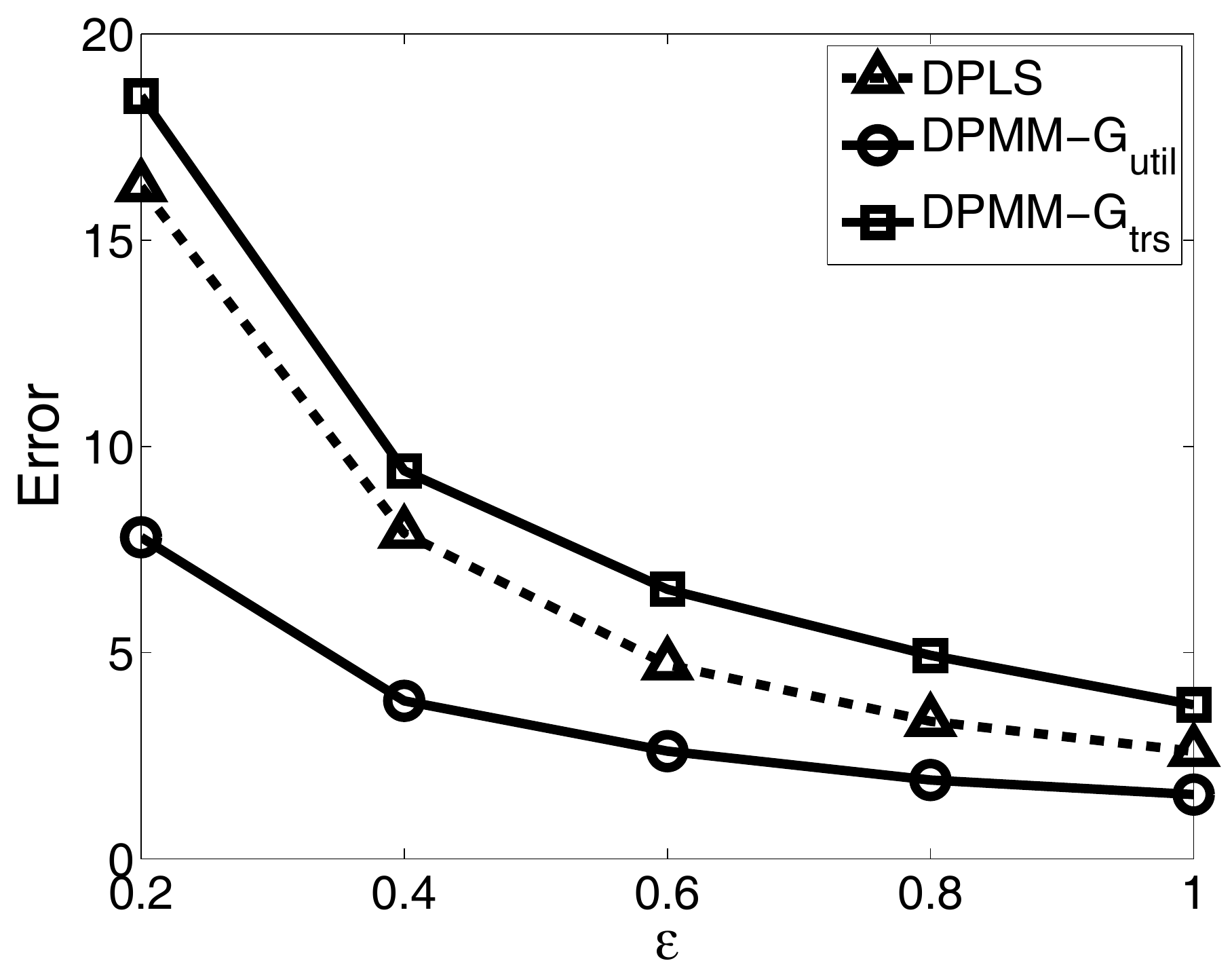}
\caption{{\small \textsc{Error} on GeoLife}}
\label{Fig-Error-eps-GeoLife}
\end{subfigure}
\begin{subfigure}{0.233\textwidth}
\centering
\includegraphics[width=4.2cm]{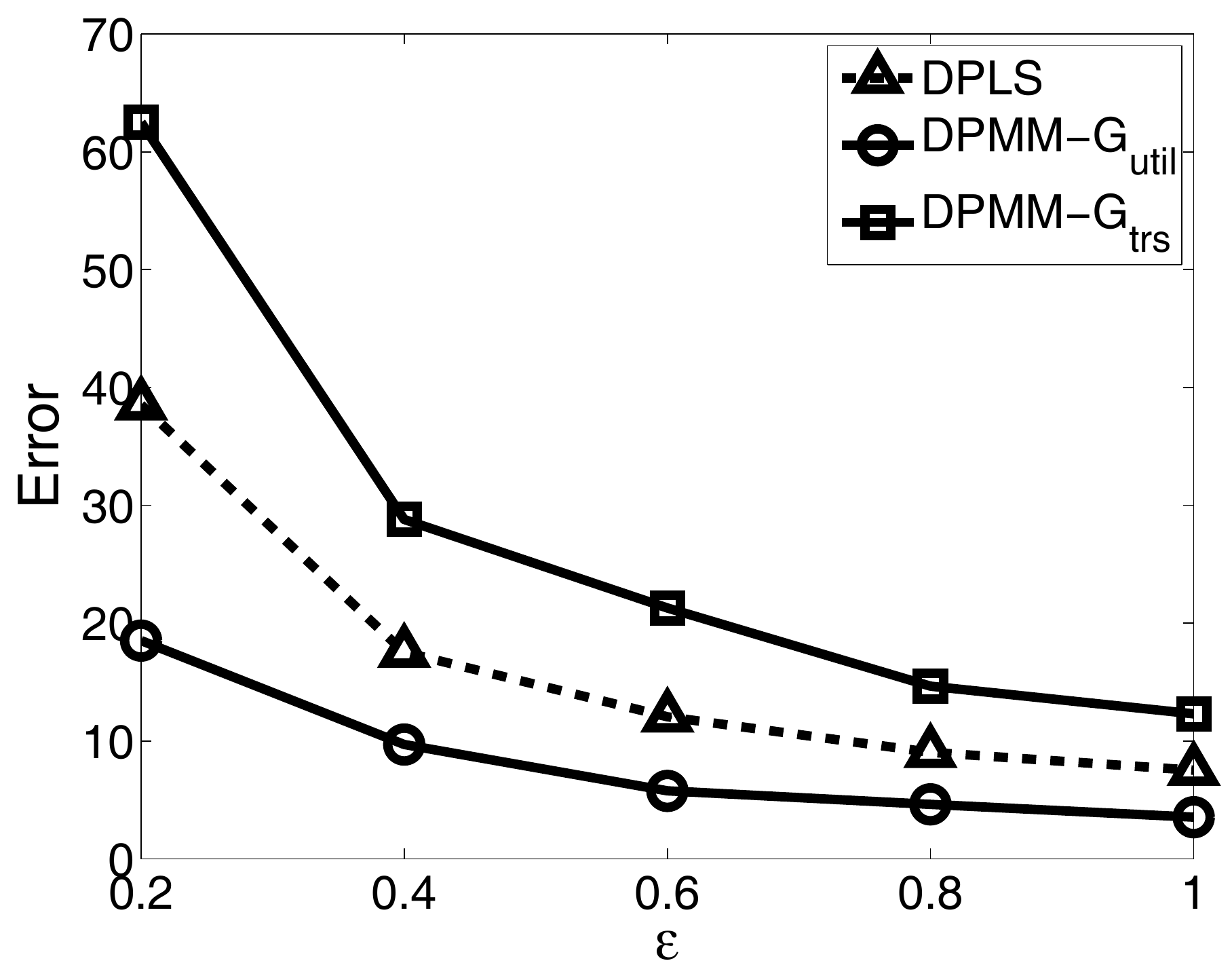}
\caption{{\small \textsc{Error} on Gowalla}}
\label{Fig-Error-eps-Gowalla}
\end{subfigure}
\caption{{\small Impact of $\epsilon$.}}
\label{Fig-performance-eps}
\end{figure}

\subsection{Performance over Time}
At each timestamp, the (smoothed) \textsc{DoP} and $\textsc{Error}$ are shown in Figure \ref{Fig-performance-time}. As expected, $G_{trs}$ provides the strongest protection of privacy, while $G_{util}$ has the lowest error on both datasets. With $G_{trs}$, the true state was protected in a set of $100$ and $70$ possible states for the two datasets. Provided such strong protection, the error also rises. With $G_{util}$, the query error was smaller than DPLS yet the $\textsc{DoP}$ was even larger than DPLS. Therefore, we can infer that customizable graph provides better trade-off between privacy and utility.

\subsection{Impact of Parameters}
We also measure the average performance over the $100$ timestamps with different parameters.

\noindent{\bf Impact of $\epsilon$.}
Figure \ref{Fig-performance-eps} reports the impact of $\epsilon$. From Figures \ref{Fig-DOP-eps-GeoLife} and \ref{Fig-DOP-eps-Gowalla}, $\textsc{DoP}$ stays the same with different $\epsilon$ because the size of $\mathcal{C}_t$ does not change with $\epsilon$. Again we see that $G_{trs}$ provides the largest $\textsc{DoP}$ with little sacrifice of utility, compared with DPLS. Figures \ref{Fig-Error-eps-GeoLife} and \ref{Fig-Error-eps-Gowalla} verifies that the larger $\epsilon$, the smaller $\textsc{Error}$, which is easy to understand because $\epsilon$ determines the shape of noise distribution.

\noindent{\bf Impact of $r$.}
To better understand the trade-off between privacy and utility with different graphs, we also tested the performance with different $G_{util}(r)$ where $r$ is the distance parameter in measurement space, as defined in Section \ref{sec-policy-graph}
 \footnote{DPLS is not affected by $r$ (not a parameter in DPLS).}.
Intuitively, with larger $r$ comes stronger protection, which is confirmed in Figures \ref{Fig-DOP-r-GeoLife} and \ref{Fig-DOP-r-Gowalla}. However, the \textsc{Error} of DPHMM is still lower than DPLS
 in most results, although it is expected that $\textsc{Error}$ grows with bigger $r$. Therefore, we can conclude that with different policy graph privacy and utility can be better tuned in different scenarios.
\begin{figure}[!t]
\centering
\begin{subfigure}{0.233\textwidth}
\centering
\includegraphics[width=4.2cm]{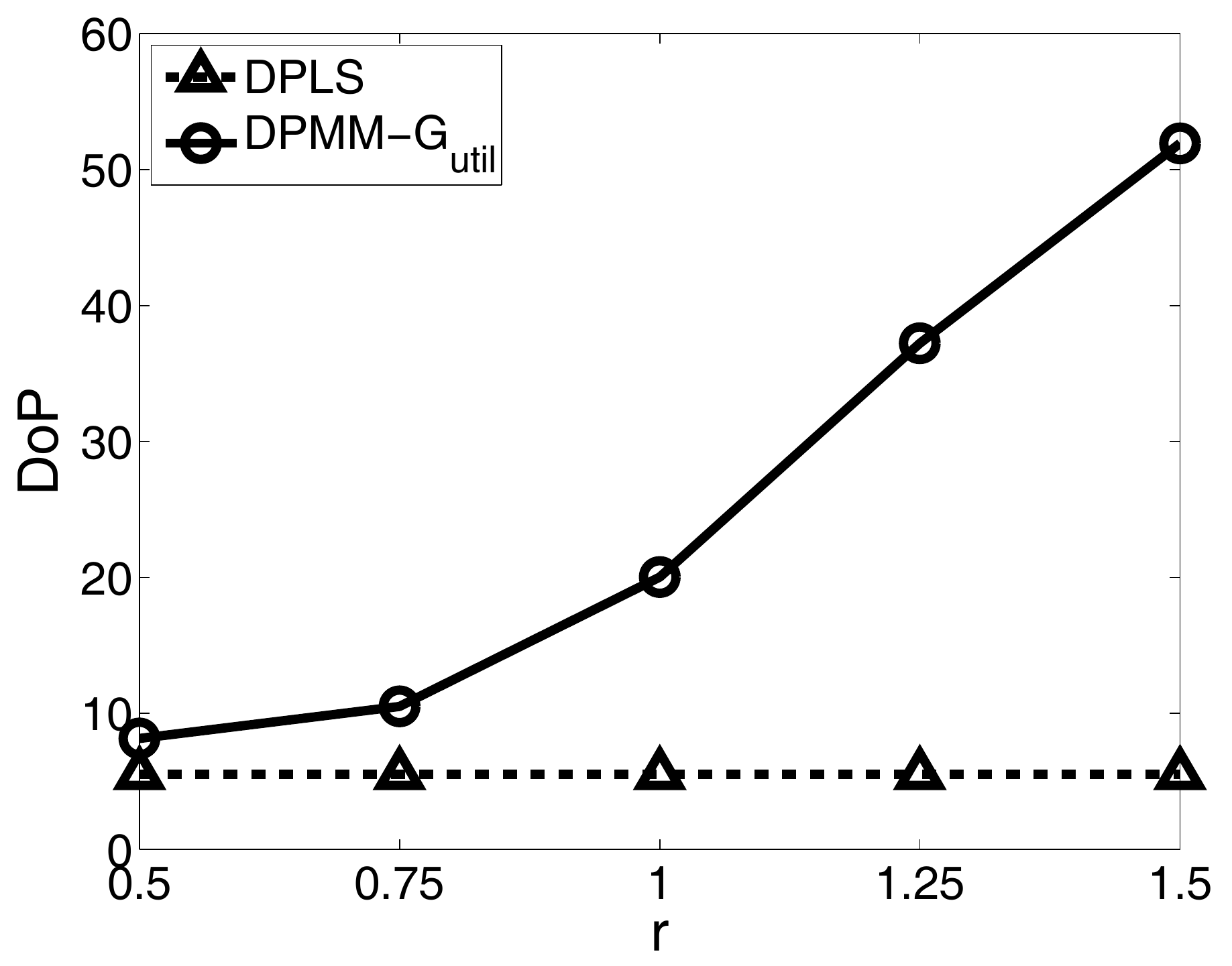}
\caption{{\small \textsc{DoP} on GeoLife}}
\label{Fig-DOP-r-GeoLife}
\end{subfigure}
\begin{subfigure}{0.233\textwidth}
\centering
\includegraphics[width=4.2cm]{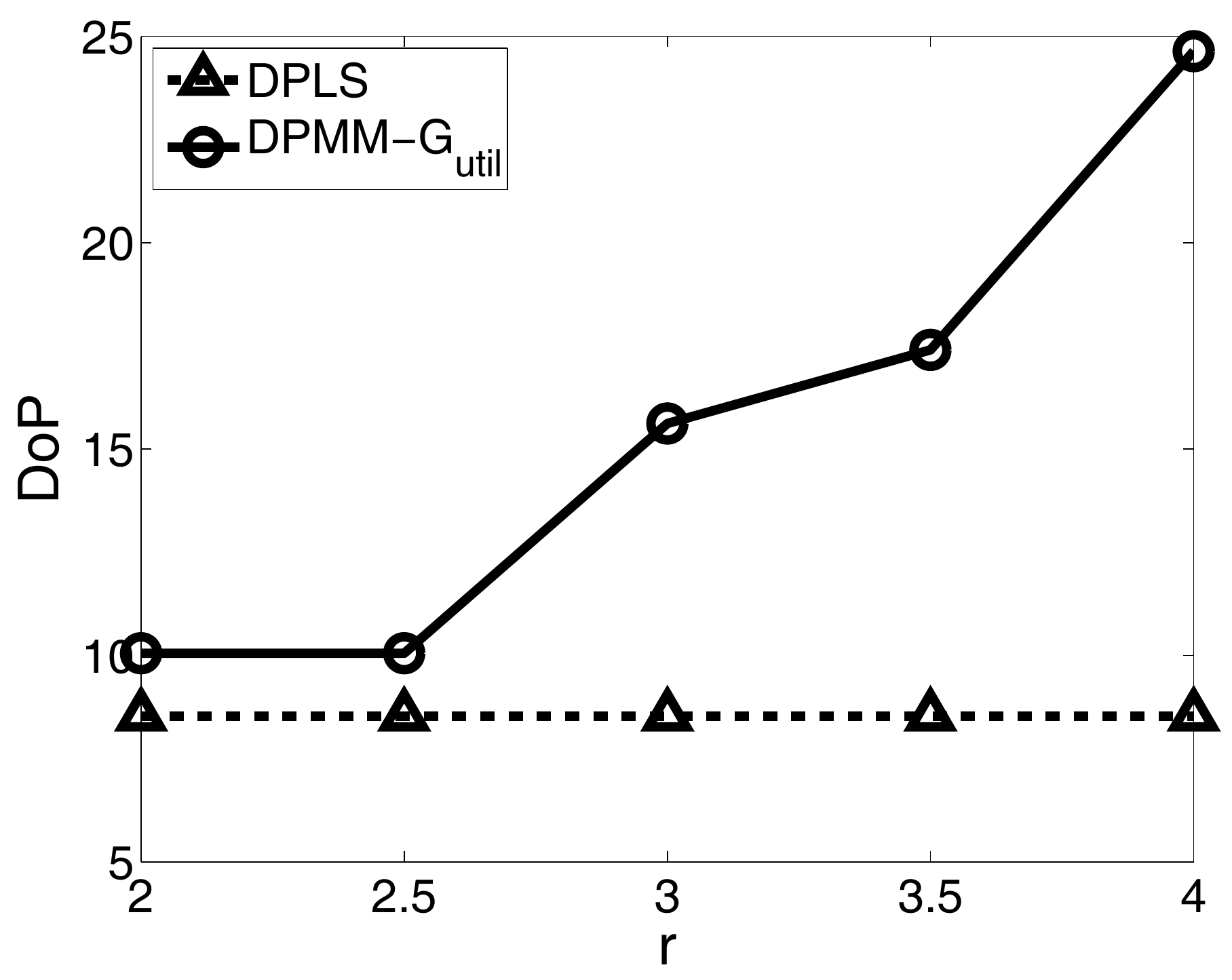}
\caption{{\small \textsc{DoP} on Gowalla}}
\label{Fig-DOP-r-Gowalla}
\end{subfigure}
\\
\begin{subfigure}{0.233\textwidth}
\centering
\includegraphics[width=4.2cm]{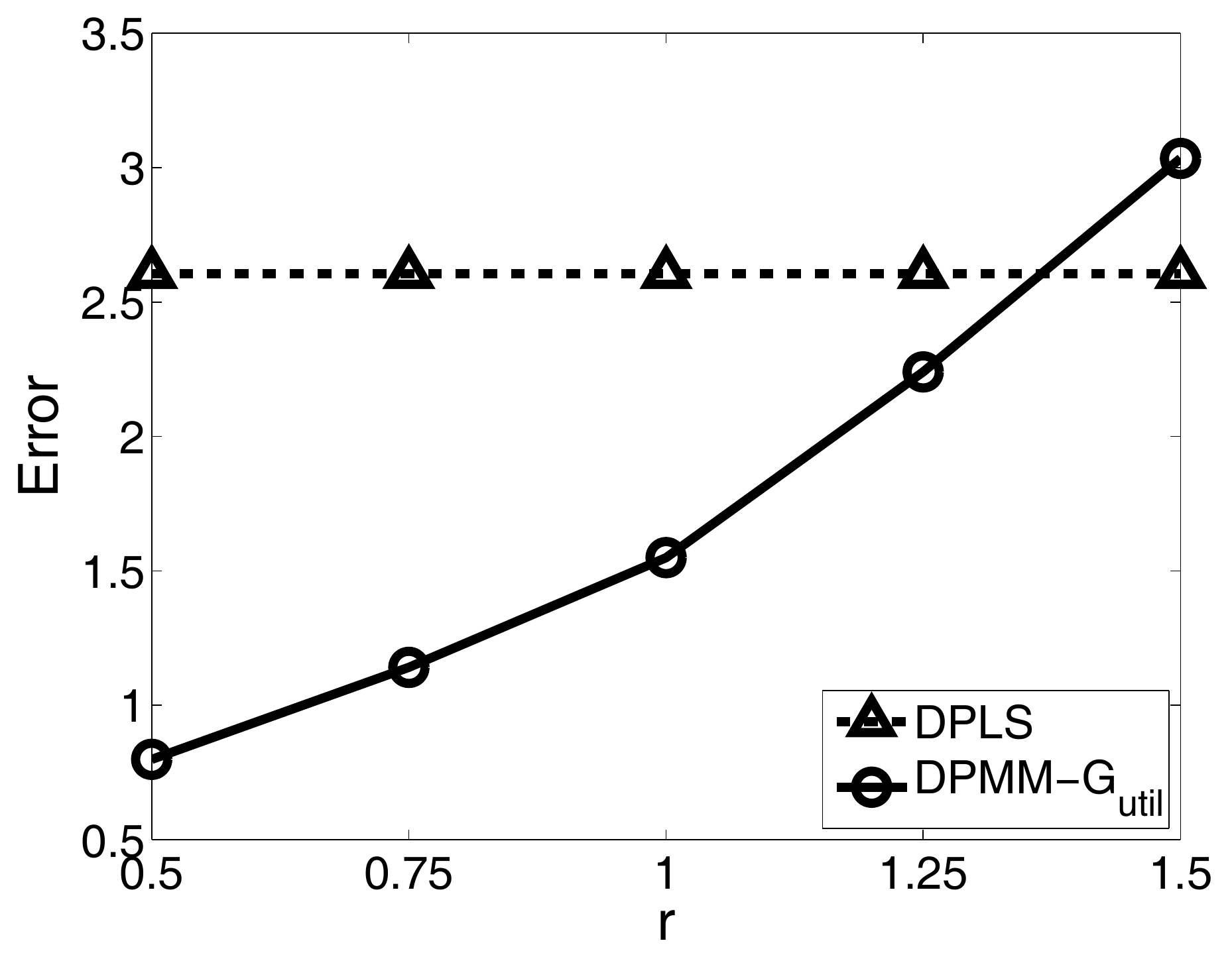}
\caption{{\small \textsc{Error} on GeoLife}}
\label{Fig-Error-r-GeoLife}
\end{subfigure}
\begin{subfigure}{0.233\textwidth}
\centering
\includegraphics[width=4.2cm]{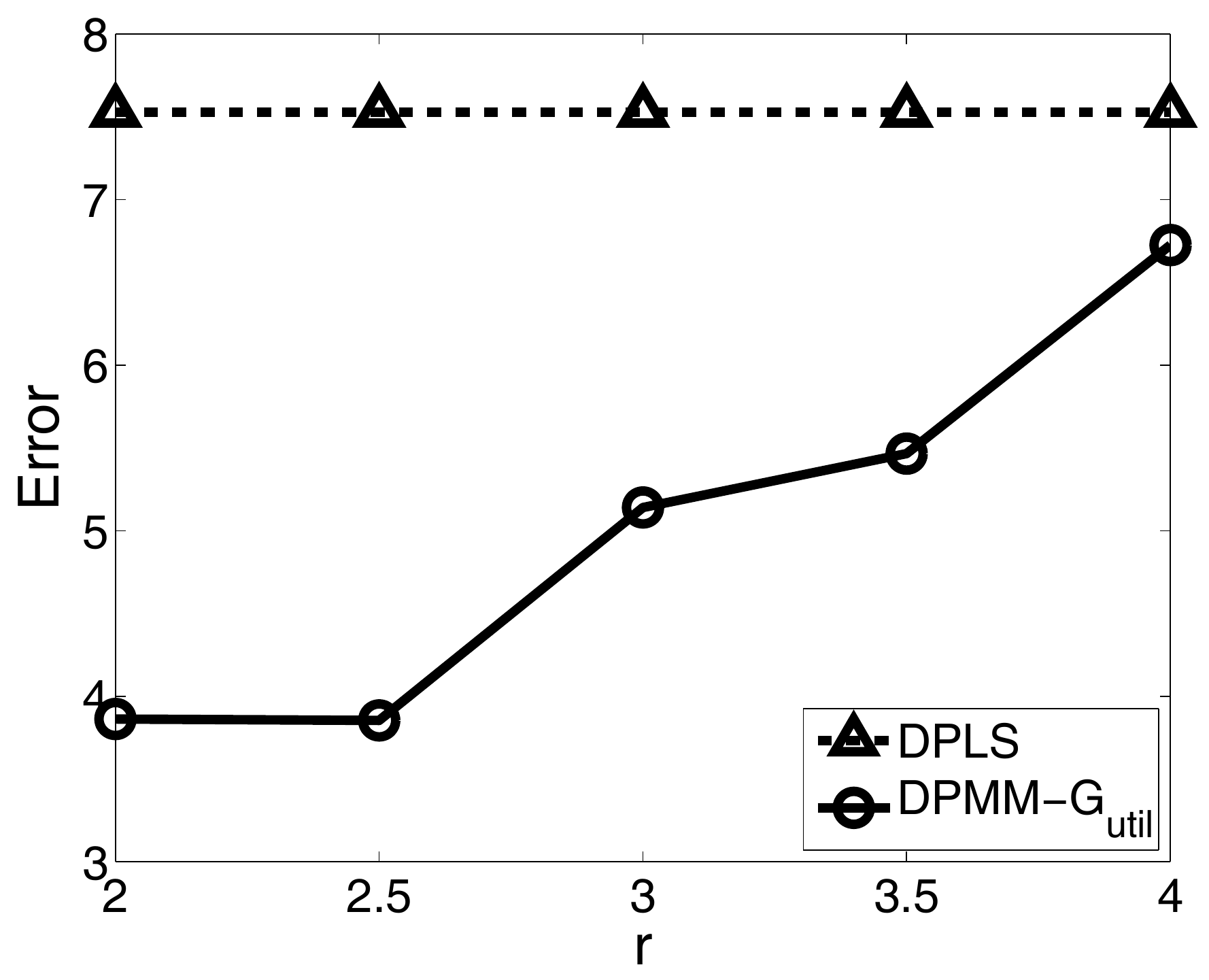}
\caption{{\small \textsc{Error} on Gowalla}}
\label{Fig-Error-r-Gowalla}
\end{subfigure}
\caption{{\small Impact of different graphs of $G_{util}$.}}
\label{Fig-impact-r}
\end{figure}

\section{Conclusion and Future Works}
In this paper we proposed DPHMM by embedding a differentially private data release mechanism in hidden Markov model. 
DPHMM guarantees that the true state in Markov model at every timestamp is protected by a customizable policy graph.
Under the temporal correlations, the graph may be reduced to subgraphs. 
Thus we studied the consequential privacy risk by introducing the notion of protectable graph based on the sensitivity hull and degree of protection.
To prevent information exposure we studied how to build an optimal protectable graph based on the current graph.
The privacy guarantee of DPHMM has also been thoroughly investigated, by comparing it with other privacy notions and studying the composition results over multiple queries and timestamps.

%

DPHMM can be used in a variety of applications to release private data for purposes like data mining or social studies. 
Future works can also study how to efficiently design and implement the policy graph for various privacy requirements.


\begin{small}
\bibliographystyle{abbrv}
\bibliography{../exportlist,ref/privacy,ref/location_privacy}
\end{small}
\newpage
\section{Appendix}

%

\subsection{Laplace Mechanism}
From the view point of $K$-norm mechanism, we can prove the following statements:
\begin{enumerate}
\item
Laplace mechanism is a special case of $K$-norm mechanism.
\item
Laplace mechanism is optimal in one-dimensional space.
\item
The $\ell_1$-norm sensitivity of Laplace mechanism contains sensitivity hull. Therefore, Laplace mechanism is not optimal in multidimensional space.
\end{enumerate}

For standard Laplace mechanism, the answer for query workload $\textbf{F}\in\mathbb{R}^{d\times N}$ \cite{Dwork-calibrating} is
\begin{align*}
\textbf{z}=\textbf{Fx}^*+\frac{S_{\textbf{F}}}{\epsilon}\tilde{\textbf{n}}
\end{align*}
where $S_{\textbf{F}}$ is the $\ell_1$-norm sensitivity of $\textbf{F}$ and $\tilde{\textbf{n}}\in\mathbb{R}^d$ are i.i.d variables from standard Laplace distribution with mean $0$ and variance $1$.
\begin{lemma}
Let $f_{\tilde{\textbf{n}}}(\tilde{\textbf{n}})$ be the joint distribution of $\tilde{\textbf{n}}\in\mathbb{R}^d$ where $\tilde{n}_1,\tilde{n}_2,\cdots,\tilde{n}_d$ are from i.i.d standard Laplace distribution. Then
\begin{align*}
f_{\tilde{\textbf{n}}}(\tilde{\textbf{n}})=\frac{1}{2^d}exp\left( -||\tilde{\textbf{n}}||_1 \right)
\end{align*}
\end{lemma}
\begin{proof}
For a scalar variable $\tilde{n}$ from standard Laplace distribution, $f_{\tilde{{n}}}(\tilde{{n}})=\frac{1}{2}exp(-|\tilde{n}|)$. Then for $\tilde{\textbf{n}}=[\tilde{n}_1,\tilde{n}_2,\cdots,\tilde{n}_d]^T$,
\begin{flalign*}
\hspace{1.2cm}
&f_{\tilde{\textbf{n}}}(\tilde{\textbf{n}}=[\tilde{n}_1,\tilde{n}_2,\cdots,\tilde{n}_d]^T)&
\\
&=\frac{1}{2}exp(-|\tilde{n}_1|)\cdot \frac{1}{2}exp(-|\tilde{n}_2|) \cdot \cdots \cdot \frac{1}{2}exp(-|\tilde{n}_d|)&
\\
&=\frac{1}{2^d}exp\left( -||\tilde{\textbf{n}}||_1 \right)&
\end{flalign*}
\end{proof}

\begin{theorem}
Let $f_\textbf{z}(\textbf{z})$ be the probability distribution of $\textbf{z}$ from standard Laplace mechanism. Then
\begin{align*}
f_\textbf{z}(\textbf{z})=\frac{\epsilon^d}{2^d S_{\textbf{F}}^d}  exp\left( -\frac{\epsilon}{S_{\textbf{F}}}\left|\left| \textbf{z}-{\textbf{F}}\textbf{x}^* \right|\right|_1  \right)
\end{align*}
\end{theorem}

\begin{theorem}
Let $K^{\Diamond}_r$ be the cross polytope $\{\textbf{x}\in\mathbb{R}^d: ||\textbf{x}||_1\leq r\}$.
Standard Laplace mechanism is a special case of $K$-norm mechanism when $K=K^{\Diamond}_{S_{\textbf{F}}}$.
\end{theorem}
\begin{proof}
In Equation (\ref{eqn-pdf-K-Norm}), let $K=K^{\Diamond}_{S_{\textbf{F}}}$. Then $\textsc{Vol}(K^{\Diamond}_{S_{\textbf{F}}})=\frac{2^d}{\Gamma(d+1)}S_{\textbf{F}}^d$.
The $K^{\Diamond}_{S_{\textbf{F}}}$-norm of any $\textbf{z}\in\mathbb{R}^d$ is $\frac{||\textbf{z}||_1}{S_{\textbf{F}}}$.
Then we can obtain
\begin{align*}
Pr(\textbf{z})=\frac{\epsilon^d}{2^d S_{\textbf{F}}^d} exp(-\frac{\epsilon}{S_{\textbf{F}}}||\textbf{z}-\textbf{Fx}^*||_1)
\end{align*}
\end{proof}

From Theorem \ref{theo-lap-K}, Statement {1} is true because $K=K^{\Diamond}_{S_{\textbf{F}}}$ in $K$-norm mechanism; Statement {2} is true because $K$ is isotropic (up to a constant) in one-dimensional space; Statement {3} is true  because $K^{\Diamond}_{S_{\textbf{F}}}$ contains the sensitivity hull.

\subsection{Details in Example \ref{example-Blowfish1}}
We explain the computation details in Example \ref{example-Blowfish1}.
\begin{example}
W.l.o.g, for a database $D$ we assume  the answer to the query in Example \ref{example-Blowfish1} is $f(D)=[10,20]^T.$ Then $f(D\cup \textbf{s}_1)=[11,20]^T$, $f(D\cup \textbf{s}_2)=[10,20]^T$, $f(D\cup \textbf{s}_3)=[11,20]^T$, $f(D\cup \textbf{s}_4)=[10,21]^T$. Given the graph in Figure \ref{Figure-example-GBlowfish}, $S_f=2=||f(D\cup \textbf{s}_3)-f(D\cup \textbf{s}_4)||_1$. Similarly,
$
\Delta f=
\pm
\left[
f(D\cup \textbf{s}_1)-f(D\cup \textbf{s}_2),
f(D\cup \textbf{s}_3)-f(D\cup \textbf{s}_4)
\right]
=
\left[
\begin{array}{cccc}
1&-1&-1&1\\
0&1&0&-1
\end{array}
\right]^T$.
The sensitivity hull $K$ is shown in Figure \ref{fig-exp-blowfish-appendix}.
For $\textbf{s}_5$, $f(D\cup \textbf{s}_5)=[11,20]^T$. Because $f(D\cup \textbf{s}_5)-f(D)=[1,0]^T$, $\textbf{s}_5$ is protected by $K$-norm mechanism for $[1,0]^T\in K$. Because $||[1,0]^T||_1=1<S_f$, it is protected by Laplace mechanism. For $\textbf{s}_6$, $f(D\cup \textbf{s}_6)-f(D)=[0,1]^T$. Thus $\textbf{s}_6$ is protected by Laplace mechanism since $||[0,1]^T||_1=1<2$. Because $[0,1]^T$ is not in $K$, it is not protected by $K$-norm mechanism.
W.l.o.g, assume $\textbf{z}=[10,21]^T$. Let $\tilde{\textbf{n}}$ be the noise injected by $K$-norm mechanism.
$\frac{Pr(\mathcal{A}(D\cup\textbf{s}_6)=\textbf{z})}{Pr(\mathcal{A}(D)=\textbf{z})}
=\frac{Pr(f(D\cup\textbf{s}_6)+\tilde{\textbf{n}}=\textbf{z})}{Pr(f(D)+\tilde{\textbf{n}}=\textbf{z})}
=\frac{Pr(\tilde{\textbf{n}}=[0,0]^T)} {Pr(\tilde{\textbf{n}}=[0,1]^T)}
=exp(\epsilon||[0,1]^T||_K-||[0,0]^T||_K)
=exp(2\epsilon)
$. Hence the unbounded DP for $\textbf{s}_6$ is $2\epsilon$. Similarly, the unbounded DP for $\{\textbf{s}_1, \textbf{s}_2, \textbf{s}_3, \textbf{s}_4, \textbf{s}_5, \textbf{s}_6\}$ are $\{\epsilon,0,\epsilon,2\epsilon,\epsilon,2\epsilon\}$ respectively. Overall, it is $2\epsilon$-unbounded-DP.
\end{example}

\begin{figure}[!htbp]
\centering
\includegraphics[width=4cm]{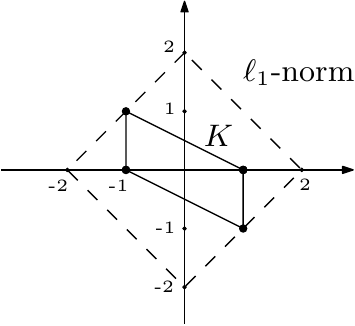}
\caption{{\small Sensitivity hull $K$ in Example \ref{example-Blowfish1}.}}
\label{fig-exp-blowfish-appendix}
\end{figure}

\subsection{Minimum Protectable Graph in $\textbf{2}$-Dimensional Space}
It is possible to design fast algorithms in low dimensional space to derive the minimum protectable graph. We propose a fast algorithm in $2$-dimensional space.

In $2$-dimensional space, it only takes $O(mlog(m))$ time to find a convex hull where $m=|\mathcal{E}|$ is the number of edges. Thus we can connect the disconnected node $\textbf{s}_i$ to the rest (at most $2m$) nodes, generating at most $2m$ convex hulls.
We use $\mathop\sum\limits_{i=1,j=i+1}^{i=h}det(\textbf{v}_i,\textbf{v}_j)$ to derive the area of a convex hull with clockwise nodes $\textbf{v}_1,\textbf{v}_2,\cdots,\textbf{v}_{h}$ where $h$ is the number of vertices and $\textbf{v}_{h+1}=\textbf{v}_1$.
By comparing the area of these convex hulls, we can find the smallest area in $O(nm^3)$ time where $n$ is the number of exposed nodes.

\begin{algorithm}[htb]
\caption{$2$D Minimum Protectable Graph}
\begin{algorithmic}[1]
\Require{
$G$, $\mathcal{C}_t$, $f:\mathcal{S}\rightarrow\mathbb{R}^2$
}
\State{$\mathcal{G}_t(\mathcal{V},\mathcal{E})\gets G\cap\mathcal{C}_t$;}
\State{$K\gets K(\mathcal{G}_t)$;}
\ForAll{exposed node $\textbf{s}_i\in\mathcal{V}$}
\State{$\textbf{s}_k\gets \emptyset$;}
\State{$min\textsc{Area}\gets \infty$;}
\ForAll{other node $\textbf{s}_j\in\mathcal{V}$}
\State{$K\gets K(\mathcal{G}_t\cup\overline{\textbf{s}_i\textbf{s}_j})$;}
\Comment{{\tt \scriptsize $O(m^2)$}}
\State{$\textsc{Area}=\mathop\sum\limits_{i=1,j=i+1}^{i=h}det(\textbf{v}_i,\textbf{v}_j)$ where $\textbf{v}_{h+1}=\textbf{v}_1$;}
\If{$\textsc{Area}<min\textsc{Area}$}
\State{$\textbf{s}_k\gets\textbf{s}_j$;}
\State{$min\textsc{Area}=\textsc{Area}$;}
\Comment{{\tt \scriptsize find minimum area}}
\EndIf
\EndFor
\State{$\mathcal{G}_t\gets\mathcal{G}_t\cup\overline{\textbf{s}_i\textbf{s}_k}$}
\State{$K\gets K(\mathcal{G}_t)$;}
\EndFor
\\\Return{graph $\mathcal{G}_t(\mathcal{V},\mathcal{E})$;}
\end{algorithmic}
\label{alg-Gt-2d-connected}
\end{algorithm}
\begin{theorem}
Algorithm \ref{alg-Gt-2d-connected} takes $O(nm^3)$ time where $m=|\mathcal{E}|$ is the number of edges and $n$ is the number of exposed nodes.
\end{theorem}

\subsection{Computing Degree of Protection}
The computation of $\textsc{DoP}$ is
to check the number of $f(\textbf{s}_j)$ inside a convex body $f(\textbf{s}_i)+K$ for all $\textbf{s}_j\in\mathcal{C}_t$.
The problem of checking whether a point is a convex body
 has been well studied in computational geometry. Thus we skip the discussion of details.
\begin{flalign}
\label{eqn-inHull-opt}
min\ { \frac{{\small 1}}{{\small 2}}}||\Delta f\cdot\textbf{x}-\textbf{v}||_2^2
\end{flalign}
\vspace{-0.55cm}
\begin{flalign*}
\hspace{2.5cm}
\textrm{subject to: }
&\textbf{1}\cdot\textbf{x}=1&\\
&\textbf{x} \succeq 0&
\end{flalign*}
where $\textbf{x}\in\mathbb{R}^{2m}$ is the unknown variable, $m=|\mathcal{E}|$ is the number of edges in $\mathcal{G}$, $\textbf{v}=f(\textbf{s}_j)-f(\textbf{s}_i)$, $\textbf{1}$ is a $1\times 2m$ vector of $[1,1,\cdots,1]$, $\textbf{x} \succeq 0$ means all elements in $\textbf{x}$ $\geq 0$. If $\Delta f\cdot\textbf{x}=\textbf{v}$ then $\textbf{s}_j$ is contained in  $f(\textbf{s}_i)+K$. Algorithm \ref{alg-exposure} summarizes the process.
\begin{algorithm}[htb]
\caption{Degree of Protection}
\begin{algorithmic}[1]
\Require{$G$, $\mathcal{C}_t$ $f$, disconnected node $\textbf{s}_i\in G\cap\mathcal{C}_t$}
\State{$\Delta f=\mathop\cup\limits_{\overline{\textbf{s}_j\textbf{s}_k}\in \mathcal{E}(G\cap\mathcal{C}_t)}
\left( f({\textbf{s}_j})-f({\textbf{s}_k}) \right)$;}
\State{$\textsc{DoP}(\textbf{s}_i)\gets 1$;}
\ForAll{$\textbf{s}_j\in\mathcal{C}_t,\textbf{s}_j\neq \textbf{s}_i$}
\State{$\textbf{v}\gets f(\textbf{s}_j)-f(\textbf{s}_i)$;}
\State{Solve $\textbf{x}$ in Equation (\ref{eqn-inHull-opt});}
\Comment{{\tt \scriptsize test $\textbf{v}\in K$}}
\If{$\Delta f\cdot\textbf{x}==\textbf{v}$}
\State{$\textsc{DoP}(\textbf{s}_i)++ $;}
\Comment{{\tt \scriptsize not exposed}}
\EndIf
\EndFor
\\\Return{$\textsc{DoP}(\textbf{s}_i)$;}
\Comment{{\tt \scriptsize if $\textsc{DoP}(\textbf{s}_i)=1$, exposed}}
\end{algorithmic}
\label{alg-exposure}
\end{algorithm}

\subsection{Attacks on Local Differential Privacy}

Why should we prevent the disclosure of unprotected nodes in a graph? In Example \ref{example-exposure}, we can see that the states $\textbf{s}_2$ and $\textbf{s}_3$ are still indistinguishable. It only matters when the true state is $\textbf{s}_5$. Following this rationale, we can also define local differential privacy
(e.g. \cite{Nissim-smooth}) based on the true state. Accordingly, this scarifies privacy for better utility.

\begin{definition}[$\epsilon$-$Local$DP]
\label{def-local-DP}
At any timestamp $t$ in MM with policy graph $G$ and true state $\textbf{s}^*_t$,
an output $\textbf{z}_t$ generated by a
 randomized algorithm $\mathcal{A}$ 
is $\epsilon$-differentially private if
for any $\textbf{z}_t$ and any states $\textbf{s}_j,\textbf{s}_k\in\mathcal{N}(\textbf{s}^*_t)\cap \mathcal{C}_t$,
$
\frac{Pr(\mathcal{A}(\textbf{s}_j)=\textbf{z}_t)}{Pr(\mathcal{A}(\textbf{s}_k)=\textbf{z}_t)}\leq e^{\epsilon}
$
 holds.
\end{definition}
Becasue a data release mechanism should be transparent to adversaries, the sensitivity hull (or $\ell_1$-norm sensitivity) should also be public to adversaries. A concern of above definition is that sensitivity hull should remain indistinguishable regardless of the true state in order to preserve privacy.
We defer such investigation to future works, with the understanding that the analysis in the rest of this paper also applies to it.

The $local$DP in Definition \ref{def-local-DP} is vulnerable to attacks using the knowledge of sensitivity hull. We use the following example to demonstrate the attack.

\begin{figure}[!htbp]
\centering
\includegraphics[width=5cm]{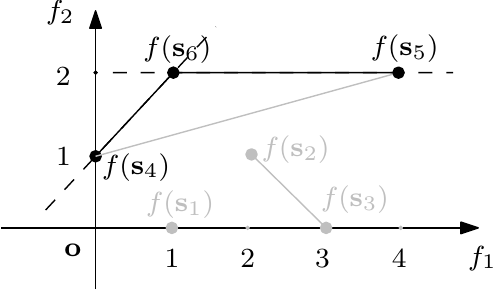}
\label{fig-localDP-attack}
\caption{{\small Attack on $Local$DP.}}
\end{figure}

\begin{example}
Continue with the running example. Assume the constraint is $\mathcal{C}_t=\{\textbf{s}_4,\textbf{s}_5,\textbf{s}_6\}$. Then we consider the instance of true state. When $\textbf{s}_t^*=\textbf{s}_4$, $K=Conv([-1,-1]^T, [1,1]^T)$. The released answer $\textbf{z}_t$ will be on the line $\overline{f(\textbf{s}_4)f(\textbf{s}_5)}$; When $\textbf{s}_t^*=\textbf{s}_6$, $\textbf{z}_t$ is on the line $\overline{f(\textbf{s}_6)f(\textbf{s}_5)}$. Then the following inference can be made:
\begin{align*}
\vast\{
\begin{array}{ll}
\textrm{If }\textbf{z}_t\in\overline{f(\textbf{s}_4)f(\textbf{s}_6)}, &\textrm{ then }\textbf{s}_t^*\neq \textbf{s}_5;\\
\textrm{If }\textbf{z}_t\in\overline{f(\textbf{s}_6)f(\textbf{s}_5)}, &\textrm{ then }\textbf{s}_t^*\neq \textbf{s}_4;\\
\textrm{If }\textbf{z}_t\notin\overline{f(\textbf{s}_4)f(\textbf{s}_6)}\cap \textbf{z}_t\notin\overline{f(\textbf{s}_6)f(\textbf{s}_5)}, &\textrm{ then }\textbf{s}_t^*= \textbf{s}_6;\\
\end{array}
\end{align*}
\end{example}

From above example, we know that the true state can be precisely figured out by attackers using the definition of sensitivity hull. Because a differentially private mechanism should be transparent to attackers, $Local$DP leaks privacy. Thus it is necessary to ensure that the sensitivity hulls remain indistinguishable for any true states.

\subsection{Adversarial Knowledge}
\label{sec-AK}
There might be a variety of adversaries with different prior knowledge in reality. Thus we consider the adversarial knowledge in this section. Similar to existing works \cite{Rastogi-adversarial-privacy,LocPriv14-arXiv}, we assume that the Markov model and the data release mechanism, including the sensitivity hull $K$, is transparent to any adversaries, meaning adversaries know how the query answers were released. If this assumption does not hold, then adversarial knowledge can only be worse, leading to less privacy disclosures.

We define constrained adversarial privacy as follows, with a similar adversary-constraint $\mathcal{C}_t^\mathcal{A}$ derived from the prior knowledge $\textbf{p}^\mathcal{A}_t$ of any adversaries:
\begin{equation*}
\mathcal{C}_t^\mathcal{A}\coloneqq \{\textbf{s}_i| \textbf{p}^{\mathcal{A}}_t[i]>0, \forall \textbf{s}_i\in \mathcal{S} \}
\end{equation*}

\begin{definition}[$\{\epsilon,\mathcal{C}_t^\mathcal{A}\}$-$Constrained$AP]
\label{def-constrainedAP}
For adversaries with knowledge $\mathcal{C}_t^\mathcal{A}$, a mechanism is $\epsilon$-adversarially private if for any output $\textbf{z}_t$ and any state $\textbf{s}_i\in\mathcal{C}_t^\mathcal{A}$,
$
\frac{Pr(\textbf{s}_i|\textbf{z}_t)}{Pr(\textbf{s}_i)}\leq e^\epsilon
$.
\end{definition}
\begin{theorem}[\cite{LocPriv14-arXiv}]
\label{theo-AP-DP-same}
If $\mathcal{C}_t=\mathcal{C}_t^\mathcal{A}$, G, $\{\epsilon,G,\mathcal{C}_t\}$-$constrained$DP (Definition \ref{def-CDP}) is equivalent to $\{\epsilon,G, \mathcal{C}_t^\mathcal{A}\}$-$constrained$AP (Definition \ref{def-constrainedAP}).
\end{theorem}
%

We discuss various adversarial knowledge  as follows.
\begin{itemize}
\item Case I: $\mathcal{C}_t^\mathcal{A}\subset\mathcal{C}_t$.
When $|\mathcal{C}_t^\mathcal{A}|=1$, the adversary has already known the true state. Then no privacy can be protected in this case. Otherwise, $G\cap\mathcal{C}_t^\mathcal{A}$ is a subgraph of $G\cap\mathcal{C}_t$.
Then 
$\{\epsilon,G,\mathcal{C}_t^\mathcal{A}\}$-DPHMM still holds.

\item Case II: $\mathcal{C}_t\subset\mathcal{C}_t^\mathcal{A}$.
Similar to the analysis in Section \ref{sec-privacy-exposure}, Algorithm \ref{alg-framework} (using $K$-norm based mechanism) may not satisfy $\{\epsilon,G,\mathcal{C}_t^\mathcal{A}\}$-DPHMM if any node in $\mathcal{C}_t^\mathcal{A}$ has $\textsc{DoP}=1$, derived from $K_t(\mathcal{G}_t)$ in Algorithm \ref{alg-framework}.

\item $\left\{\left(\mathop{max}\limits_{\forall \textbf{s}_i,\textbf{s}_j\in\mathcal{C}_t^\mathcal{A}\cap\mathcal{C}_t}||f(\textbf{s}_i)-f(\textbf{s}_j)||_{K_t}\right)\epsilon,\mathcal{C}_t^\mathcal{A}\cap\mathcal{C}_t\right\}$-$cons$-$\\trained$AP holds in both cases. 
    Hence \\
    $\left\{\left(\mathop{max}\limits_{\forall \textbf{s}_i,\textbf{s}_j\in\mathcal{C}_t^\mathcal{A}\cap\mathcal{C}_t}||f(\textbf{s}_i)-f(\textbf{s}_j)||_{K_t}\right)\epsilon,\mathcal{C}_t^\mathcal{A}\cap\mathcal{C}_t\right\}$-$cons$-$\\trained$DP
     also holds by Theorem \ref{theo-AP-DP-same}.
\end{itemize}

\end{document}